\documentclass[a4paper,10pt,reqno]{amsart}
\usepackage[centertags]{amsmath}
\usepackage{amsfonts,amssymb,amsthm}
\usepackage{float,longtable}
\usepackage{graphicx}
\usepackage{a4wide}
\usepackage[latin1]{inputenc}
\usepackage{newlfont}
 \theoremstyle{plain}
 \begingroup
 \newtheorem{thm}{Theorem}[section]

 \endgroup

 \theoremstyle{definition}
 \begingroup

 \endgroup

 \theoremstyle{remark}
 \begingroup
 \newtheorem{rem}[thm]{Remark}
 \endgroup

 \numberwithin{equation}{section}

\newcommand{\norm}[1]{\left\Vert#1\right\Vert}
\newcommand{\snorm}[1]{\left|#1\right|}

\title[3D-1D asymptotic analysis for thin curved domains in nonlinear elasticity]
{Asymptotic models for curved rods derived from nonlinear elasticity by $\Gamma$-convergence}
\author[]{LUCIA SCARDIA}
\address[]{S.I.S.S.A., Via Beirut 2-4, 34014, Trieste, Italy}
\email[]{scardia@sissa.it}

\begin{document}
\maketitle
\begin{center}
\begin{minipage}{12cm}
\small{
\noindent {\bf Abstract.}  We study the problem of the rigorous derivation of one-dimensional 
models for a thin curved beam starting from three-dimensional nonlinear elasticity. 
We describe the limiting models obtained for different scalings of the energy. 
In particular, we prove that the limit functional corresponding to higher scalings 
coincides with the one derived by dimension reduction starting from linearized 
elasticity. 

\vspace{15pt}

\vspace{6pt}
}
\end{minipage}
\end{center}

\bigskip

\section{Introduction}

One of the main problems in nonlinear elasticity is to understand the relation between the 
three-dimensional theory and lower dimensional models for thin structures. In the classical 
approach these theories are usually deduced via formal asymptotic expansions or adding extra 
assumptions on the kinematics of the three-dimensional deformations (see, e.g., \cite{Cia2}). 
Recently the problem of the rigorous derivation of lower dimensional theories has been studied 
using a variational approach, which is based on the analysis of the limit of the 3D elastic 
energy in the sense of $\Gamma$-convergence. 

The first result in this direction is due to E. Acerbi, G. Buttazzo and D. Percivale 
(see \cite{ABP91}), who deduced a nonlinear model for elastic strings by means of a 3D-1D 
dimension reduction. The two-dimensional analogue was studied by H. Le Dret and A. Raoult 
who derived a nonlinear model for planar membranes (see \cite{LDR95}) and for shell membranes 
(see \cite{LDR00}). The more delicate case of plates was justified more recently by G. Friesecke, R.D. James and 
S. M\"uller in \cite{FJM02}, while the case of shells was treated in \cite{FJMM03}.
For a complete survey on plate theories we refer to \cite{FJM06}. 

Concerning the derivation of one-dimensional models, the study of straight rods in 
the nonlinear case has been performed by M.G. Mora, S. M\"uller (see \cite{MM03}, \cite{MGMM04}) and, 
independently, by O. Pantz (see \cite{P02}). 
In all the previous results, the different limiting models correspond to different scalings of the 3D 
energy in terms of the parameter describing the thickness in the case of plates and the diameter of the 
cross-section in the case of rods.

In this paper we deal with the case of a thin \textit{curved} heterogeneous beam made of a hyperelastic 
material. This is a sequel to a previous work, where lower scalings of the energy were considered (see \cite{S06}).

In the following we shall denote by $\Omega$ the set $(0,L)\times D$, where $L>0$ and $D$ is a bounded Lipschitz 
domain in $\mathbb{R}^2$. Given $h>0$, we shall consider a beam, whose reference configuration is given by
$$\widetilde{\Omega}_{h}:= \{\gamma(s) + h\,\xi\,\nu_{2}(s) + h\,\zeta\, \nu_{3}(s) : (s,\xi,\zeta)\in \Omega\},$$
where $\gamma:(0,L)\to \mathbb{R}^{3}$ is a smooth simple curve describing the mid-fiber of the beam, and $\nu_2, \nu_3:(0,L)\to \mathbb{R}^{3}$ are two smooth vectors such that $\{\gamma',\nu_2, \nu_3 \}$ provides an orthonormal frame along the curve. 
We denote with $\tau$ the unit vector $\gamma'$ tangent to the curve $\gamma$. 
We notice that the cross-section of the beam is constant along $\gamma$ and is given by the set $hD$.
A natural parametrization of $\widetilde{\Omega}_{h}$ is given by
\begin{equation*}
\Psi^{(h)} :  \Omega \rightarrow \widetilde{\Omega}_{h}, \quad (s,\xi,\zeta)\mapsto \gamma(s) + h\,\xi\,\nu_{2}(s) + h\,\zeta\,\nu_{3}(s),
\end{equation*}
which is one-to-one for $h$ small enough.

The starting point of the variational approach is the elastic energy per unit cross-section
\begin{equation*}\label{puv}
\tilde{I}^{(h)}(\tilde{y}):= \frac{1}{h^2}\int_{\widetilde{\Omega}_{h}}
W\big(\big(\Psi^{(h)}\big)^{-1} (x),\nabla \tilde{y}(x)\big) dx
\end{equation*}
of a deformation $\tilde{y} \in W^{1,2}(\widetilde{\Omega}_{h};\mathbb{R}^{3})$. The stored 
energy density $W:\Omega\times {\mathbb M}^{3\times 3}\to [0,+\infty]$ is required to satisfy
some natural properties:
\smallskip
\begin{itemize}
\item $W$ is frame indifferent: 
$W(z,RF) = W(z,F)$ for a.e.\ 
$z\,\in \Omega$, every $F\in \mathbb{M}^{3\times 3}$, and every $R\in SO(3)$;
\smallskip
\item $W(z,F)\geq C\,\mbox{dist}^{2}(F,SO(3))$  for a.e.\  $z\in \Omega$ and
every $F\in\mathbb{M}^{3\times 3}$;
\smallskip
\item $W(z,R)=0$ for a.e.\ $z\in\Omega$ and every $R\in SO(3)$. 
\end{itemize}
\smallskip
For the complete list of assumptions on $W$ we refer to Section~2.

The goal is to provide a description of the asymptotic behaviour of $\tilde{I}^{(h)}$, 
as $h\rightarrow 0$, by means of $\Gamma$-convergence  (see \cite{DM93} for a comprehensive
introduction to $\Gamma$-convergence). 

In \cite{S06} we studied the case of energies $\tilde{I}^{(h)}$ of order $h^\beta$ with $\beta \in [0,2]$. 
We proved that, as for straight beams, energies of order $1$ correspond to stretching and shearing 
deformations, leading to a \textit{string theory} as $\Gamma$-limit, while energies of order $h^2$ correspond 
to bending flexures and torsions keeping the mid-fiber unextended, leading to a \textit{rod theory} 
as $\Gamma$-limit. This last result has been obtained also by P. Seppecher and C. Pideri in \cite{Sepp2}, 
independently. Finally, in \cite{S06} it was also shown that the $\Gamma$-limit of $h^{- \beta}\tilde{I}^{(h)}$ 
with $\beta \in (0,2)$ provides a degenerate model.

In this paper we consider the scalings $h^{\beta}$ with $\beta > 2$. More precisely, we prove that if the 
energy $\tilde{I}^{(h)}$ is of order $h^4$, then the corresponding relevant deformations are close to a 
rigid motion, so that the $\Gamma$-limit describes a partially linearized model. 
This result generalizes to the case of curved rods what was proved in \cite{MGMM04} for straight rods. 
Furthermore, we show that the scalings $\beta > 4$ lead to the linearized theory for rods, while the 
scalings $\beta \in (2,4)$ correspond to a constrained linearized theory which is the analogous, 
in the one-dimensional case, 
of the von K\'arm\'an theory for plates (see \cite{FJM06}).  We want to underline that for the intermediate 
scalings $\beta \in (2,4)$ a more delicate analysis is required in order to prove the result, as in the 
corresponding case in \cite{FJM06}.

We first present a compactness result for sequences of deformations having equibounded energies 
$h^{-\beta}\tilde{I}^{(h)}$ with $\beta > 2$ (Theorem \ref{compactness}). More precisely, we prove that 
if $\tilde{I}^{(h)}\big(\tilde{y}^{(h)}\big) \leq c\,h^{\beta}$, $\beta > 2$, then there exist some constants $\bar{R}^{(h)}\in SO(3)$ such that $\bar{R}^{(h)} \rightarrow \bar{R}$ and, up to subsequences,
\begin{equation*}
\nabla\big(\big(\bar{R}^{(h)}\big)^{T}\,\tilde{y}^{(h)}\big)\circ \Psi^{(h)} \rightarrow Id \quad \hbox{strongly in } L^{2}(\Omega;\mathbb{M}^{3\times 3}).
\end{equation*}
In other words, up to a rigid motion, the deformations $\tilde{y}^{(h)}$ converge to the identity. This 
naturally leads to introduce a new sequence of scaled deformations $Y^{(h)}$, given by $\big(\bar{R}^{(h)}\big)^{T}\tilde{y}^{(h)}\circ \Psi^{(h)}$ (up to an additive constant) and to study the 
deviation of $Y^{(h)}$ from $\Psi^{(h)}$. To this aim, we define the scaled averaged displacement 
\begin{equation*}
v^{(h)}(s) := \,\frac{1}{h^{(\beta - 2)/2}}\int_{D}\big(Y^{(h)}(s,\xi,\zeta) - \Psi^{(h)}(s,\xi,\zeta)\big)\,d\xi\,d\zeta
\end{equation*}
and the twist angle of the cross-section
\begin{equation*}
w^{(h)}(s) := \,\frac{1}{h^{\beta/2}}\left(\frac{1}{\mu(D)}\int_{D}\big(Y^{(h)}(s,\xi,\zeta) - \Psi^{(h)}(s,\xi,\zeta)\big)\cdot (\xi\,\nu_{3}(s) - \zeta\,\nu_{2}(s))\,\,d\xi\,d\zeta\right),
\end{equation*}
where $\mu(D) := \int_{D}\big(\xi^2 + \zeta^2\big)\,d\xi\,d\zeta$. Finally, we introduce a function $u^{(h)}$, which measures the extension of the mid-fiber and is given by
\begin{equation*}
u^{(h)}(s) := \left\{
\begin{array}{ll}
\vspace{.15cm}
\displaystyle\frac{1}{h^{\beta - 2}}\int_{s_h}^{s} \Big(\int_{D} \partial_s\big(Y^{(h)}(s,\xi,\zeta) - \Psi^{(h)}(s,\xi,\zeta)\big)\cdot \tau(\sigma) \,\,d\xi\,d\zeta\Big)\,d\sigma  \,\, \mbox{if }\, 2<\beta<4,\\
\displaystyle\frac{1}{h^{\beta/2}}\int_{s_h}^{s} \Big(\int_{D} \partial_s\big(Y^{(h)}(s,\xi,\zeta) - \Psi^{(h)}(s,\xi,\zeta)\big)\cdot \tau(\sigma) \,\,d\xi\,d\zeta\Big)\,d\sigma  \,\quad\, \mbox{if }\, \beta\geq 4,
\end{array}
\right.\\
\end{equation*}
where $s_h \in (0,L)$ is chosen in such a way that $u^{(h)}$ has zero average on $(0,L)$.

In Theorem \ref{compactness} it is then shown that, up to subsequences, the following convergence 
properties are satisfied: 
\vspace{.2cm}
\begin{itemize}
\item $v^{(h)}\rightarrow v$ \, strongly in $W^{1,2}((0, L);\mathbb{R}^{3})$, for some
$v\in W^{2,2}((0, L);\mathbb{R}^{3})$ with $v'\cdot \tau = 0$;
\vspace{.15cm}
\item $w^{(h)}\rightharpoonup w$ \, weakly in $W^{1,2}(0, L)$, for some $w\in W^{1,2}(0,L)$;
\vspace{.15cm}
\item $u^{(h)}\rightharpoonup u$ \, weakly in $W^{1,2}(0, L)$, for some $u\in W^{1,2}(0,L)$.
\vspace{.15cm}
\end{itemize}

In Theorems \ref{scithm2} and \ref{bfa} the $\Gamma$-limit of the functionals 
$\big(\tilde{I}^{(h)}/h^\beta\big)$, for $\beta \geq 4$, is identified. 

\noindent
In the case $\beta = 4$ we  show that it is an integral functional depending on $u$, $v$ and $w$, of the form 
\begin{equation*}
I_{4}(u,v,w):= \frac{1}{2}\int_{0}^{L}Q^{0}\Big(s, u' +
\frac{1}{2}\big((v'\cdot\nu_2)^2 + (v'\cdot \nu_3)^2\big),B' + 2\,\hbox{skw}\big(R_0^T R'_0 B\big)
\Big)\,ds,
\end{equation*}
where $B\in W^{1,2}((0,L);\mathbb{M}^{3\times 3})$ denotes the
matrix
\begin{equation}\label{1A}
B := \left(\begin{array}{c}
0\\
v'\cdot\nu_2\\
v'\cdot\nu_3
\end{array}
\begin{array}{c}
- v'\cdot\nu_2\\
0\\
w
\end{array}
\begin{array}{c}
- v'\cdot\nu_3\\
- w\\
0
\end{array} \right)
\end{equation}
and $Q^{0}$ is a quadratic form arising from a minimization problem involving the quadratic form 
of linearized elasticity (see (\ref{defQ0})).

If $\beta > 4$ the limit functional is fully linearized and it is given by 
\begin{equation*}
I_\beta(u,v,w):= \frac{1}{2}\int_{0}^{L}Q^{0}\big(s,u',B' + 2\,\hbox{skw}\big(R_0^T R'_0 B\big)\big)\,ds,
\end{equation*}
where $B$ and $Q^{0}$ are defined as before. We notice that $I_\beta$ coincides with the functional obtained by dimension reduction starting from linearized elasticity (see Remark \ref{comparison}).

Finally, in the case $\beta \in (2,4)$, it turns out that $v$ and $u$ are linked by the following nonlinear constraint: 
\begin{equation}\label{con}
u' = -\,\frac{1}{2}\big((v'\cdot\nu_2)^2 + (v'\cdot \nu_3)^2\big).
\end{equation}
Therefore, the function $u$ is completely determined, once $v$ is known, and hence the limit functional depends on $v$ and $w$ only. More precisely, it is given by
\begin{equation*}
I_\beta(v,w):= \frac{1}{2}\int_{0}^{L}Q\big(s, B' + 2\,\hbox{skw}\big(R_0^T R'_0 B\big)\big)\,ds,
\end{equation*}
where $B$ is defined as in (\ref{1A}) and $Q$ is obtained by minimizing the quadratic form $Q^{0}$ with respect to its second argument (see (\ref{defQ})).

The last section of the paper is devoted to the extension of the previous results to the case of a 
thin ring. In other words, the mid-fiber of the beam is assumed to be a closed curve in $\mathbb{R}^3$. 
We prove that in this case the limiting functionals are finite only on the class of triples $(u,v,w)$ such
that $v$ and $w$ satisfy the periodic boundary conditions $v(0) = v(L)$ and $w(0) = w(L)$ (see Theorem \ref{compbc}).
Moreover, on this class the $\Gamma$-limits coincide with the previous functionals $I_{\beta}$ (see Theorem \ref{limite}).

\section{Notations and formulation of the problem}
In this section we describe the geometry of the unstressed curved beam. 
Let $\gamma : [0, L] \rightarrow \mathbb{R}^{3}$ be a simple regular curve of class $C^{3}$ parametrized by the arc-length and let 
$\tau = \gamma'$ be its unit tangent vector. We assume that there exists an orthonormal frame of class $C^{2}$ along the curve. More precisely, we assume that there exists $R_{0}\in C^{2}([0, L]; \mathbb{M}^{3\times 3})$ such that $R_{0}(s)\in SO(3)$ for every $s \in [0, L]$ and $R_{0}(s)\,e_{1} = \tau(s)$ for every $s \in [0, L]$, where $e_i$, for $i=1,2,3$, denotes the $i$-th vector of the canonical basis of $\mathbb{R}^{3}$ and \mbox{$SO(3) = \big\{R\in\mathbb{M}^{3\times 3} : R^T R = Id,\, \det R = 1 \big\}$}.
We set 
$$\nu_k (s) := R_{0}(s)\,e_{k}, \, \hbox{for } k = 2,3.$$
We can introduce three scalar functions $\varrho$, $k_2$ and $k_3$ in $C^{1}([0, L])$ such that
\begin{align}\label{curv}
\tau'(s) =&\, k_2(s)\,\nu_2(s) + k_3(s)\,\nu_3(s),\nonumber\\
\nu_2'(s) =&\, - k_2(s)\,\tau(s) + \varrho(s)\,\nu_3(s),\nonumber\\
\nu_3'(s) =&\, - k_3(s)\,\tau(s) - \varrho(s)\,\nu_2(s). 
\end{align}
Note that the curvature of $\gamma$ can be easily recognized as $\sqrt{k_2^2 + k_3^2}$ and the torsion 
of $\gamma$ as $\varrho + \dfrac{k_2 k_3' - k_3 k_2'}{k_2^2 + k_3^2}$.

Let $D\subset \mathbb{R}^{2}$ be a bounded open connected set with Lipschitz boundary such that
\begin{equation}\label{dom1}
\int_{D}\xi\,\zeta\, d\xi\, d\zeta = 0
\end{equation} 
and
\begin{equation}\label{dom2}
\int_{D}\xi\,d\xi\,d\zeta = \int_{D}\zeta\,d\xi\,d\zeta = 0,
\end{equation} 
where $(\xi,\zeta)$ stands for the coordinates of a generic point of $D$. Without loss of generality, we can also assume $\mathcal{L}^2(D) = 1$. We set $\Omega:= (0, L)\times D$.
\newline
The reference configuration of the thin beam is given by
$$\widetilde{\Omega}_{h}:= \{\gamma(s) + h\,\xi\,\nu_{2}(s) + h\,\zeta\, \nu_{3}(s) : (s,\xi,\zeta)\in \,\Omega\},$$
where $h$ is a small positive parameter. Clearly the curve $\gamma$ and the set $D$ represent the mid-fiber and the cross-section of the beam, respectively.
The set $\widetilde{\Omega}_{h}$ is parametrized by the $C^2$ map
\begin{equation*}
\Psi^{(h)} : \Omega \rightarrow \widetilde{\Omega}_{h}\,: \quad (s,\xi,\zeta)\mapsto \gamma(s) + h\,\xi\,\nu_{2}(s) + h\,\zeta\,\nu_{3}(s),
\end{equation*}
which is one-to-one for $h$ small enough.

We assume that the thin beam is made of a hyperelastic material whose stored energy density 
$W : \Omega\times\mathbb{M}^{3\times 3} \rightarrow [0, + \infty]$ is a Carath\'eodory function
satisfying the following hypotheses:
\begin{itemize}
\vspace{.2cm}
\item[(i)] there exists $\delta>0$ such that the function $F\mapsto W(z,F)$ is of class $C^{2}$ on the set\\$\big\{F\in\mathbb{M}^{3\times 3}: \mbox{dist}(F,SO(3)) < \delta\big\}$  for a.e. $z\,\in \Omega$;
\vspace{.2cm}
\item[(ii)] the second derivative $\partial^{2}W/\partial F^{2}$ is a Carath\'eodory function on the set
\begin{equation*}
\Omega\times\{F\in \mathbb{M}^{3\times 3}:\,\mbox{dist}(F,SO(3)) < \delta \}
\end{equation*}
and there exists a constant $C_{1} > 0$ such that
\begin{align*}
&\bigg|\frac{\partial^{2}W}{\partial F^{2}}(z,F)[G,G]\bigg| \leq C_{1} |\,G\,|^{2}\,\, \mbox{for a.e. } z\in \Omega,\,\, \mbox{every } F \, \hbox{with }\, \hbox{dist}(F,SO(3))<\delta\\
& \mbox{and every } G\in \mathbb{M}^{3\times 3}_{sym};
\end{align*}
\item[(iii)] $W$ is frame indifferent, i.e., $W(z,RF) = W(z,F)$  for a.e. $z\,\in \Omega$, every $F\in \mathbb{M}^{3\times 3}$ and every $R\in SO(3)$;
\vspace{.2cm}
\item[(iv)] $W(z,R)=0$ for every  $R\in SO(3)$;
\vspace{.2cm}
\item[(v)] $\exists$ $C_{2} >\,0$ independent of $z$ such that $W(z,F)\geq C_{2}\,
\mbox{dist}^{2}(F,SO(3))$  for a.e. $z\in \Omega$ and every  $F\in\mathbb{M}^{3\times 3}$. 
\end{itemize}
Notice that, since we do not require any growth condition from above, $W$ is allowed to assume the value $+ \infty$ outside a neighborhood of the set (\ref{set}). Therefore our treatment covers the physically relevant case in which $W = + \infty$ for $\det F < 0$, $W\rightarrow + \infty$ as $\det F \rightarrow 0^+$.

Let $\tilde{y} \in W^{1,2}(\widetilde{\Omega}_{h};\mathbb{R}^{3})$\, be a deformation of $\widetilde{\Omega}_{h}$. The elastic energy per unit cross-section associated to $\tilde{y}$ is defined by
$$\tilde{I}^{(h)}(\tilde{y}):= \frac{1}{h^{2}}\int_{\widetilde{\Omega}_{h}}
W\big(\big(\Psi^{(h)}\big)^{-1} (x),\nabla
\tilde{y}(x)\big) dx.$$

\noindent
We conclude this section by analysing some properties of the map
$\Psi^{(h)}$, which will be useful in
the sequel. We will use the following notation: for any function $z\in W^{1,2}(\Omega;\mathbb{R}^3)$ we set
\begin{equation*}
\nabla_{h}z := \left(\partial_s z\,\Big|\,\frac{1}{h}\,\partial_\xi z\,\Big|\, \frac{1}{h}\,\partial_\zeta z\right).
\end{equation*}

We observe that $\nabla_{h}\Psi^{(h)}$
can be written as the sum of the rotation $R_{0}$ and a
perturbation of order $h$, that is,
\begin{equation}\label{psi}
\nabla_{h}\Psi^{(h)}(s,\xi,\zeta) = R_{0}(s) + h\,\left(\xi\,\nu'_{2}(s) + \zeta\,\nu'_{3}(s)\right)\otimes e_1.
\end{equation}
From this fact it follows that, as $h\rightarrow 0$,
\begin{equation}\label{convdet}
\nabla_{h}\Psi^{(h)}(s,\xi,\zeta)\rightarrow  R_{0}(s)\quad \mbox{and}\quad\det \big(\nabla_{h}\Psi^{(h)}\big) \rightarrow 1 = \det R_{0} \,\,\,\mbox{uniformly}.
\end{equation}
This implies that for $h$ small enough
$\nabla_{h}\Psi^{(h)}$ is invertible at each
point of $\Omega$. Since the inverse of
$\nabla_{h}\Psi^{(h)}$ can be written as
\begin{equation}\label{invA}
\big(\nabla_{h}\Psi^{(h)}\big)^{-1}(s,\xi,\zeta) =
R_{0}^{T}(s) - h\,R_{0}^{T}(s)\,\big[\left(\xi\,\nu'_{2}(s) +
\zeta\,\nu'_{3}(s)\right)\otimes \tau(s)\big] + O(h^{2})
\end{equation}
with \, $O(h^{2})/h^{2}$ \, uniformly bounded,
$\big(\nabla_{h}\Psi^{(h)}\big)^{-1}$ converges
to $R_{0}^{T}$ uniformly.

\section{Compactness results}

In this section we study the compactness properties of sequences of
deformations having energy $\tilde{I}^{(h)}$ of order $h^{\beta}$ with $\beta > 2$. 
For notational convenience we prefer to write $\beta > 2$ as $2\alpha -2$ with $\alpha > 2$.
The main ingredient in the proof is the following rigidity result, proved by G. Friesecke, R.D.
James and S. M\"{u}ller in \cite{FJM02}.
\begin{thm}\label{Teorigid}
Let $U$ be a bounded Lipschitz domain in $\mathbb{R}^{n}$, $n\geq
2$. Then there exists a constant $C(U)$ with the following
property: for every $v\in W^{1,2}(U;\mathbb{R}^{n})$ there is an
associated rotation $R\in SO(n)$ such that
\begin{equation*}
\norm{\nabla v - R}_{L^{2}(U)} \leq C(U)\norm{\textnormal{dist}(\nabla
v, SO(n))}_{L^{2}(U)}.
\end{equation*}
\end{thm}
\begin{rem}
The constant $C(U)$ is invariant under uniform scaling of the domain; moreover it 
can be chosen independent of $U$ for a family of sets that are Bilipschitz images of 
a cube (with uniform Lipschitz constants), as remarked in \cite{FJMM03}.
\end{rem}

Before stating the compactness theorem, let us introduce some sequences which will
be widely used in the sequel. Given a sequence of deformations $Y^{(h)}: \Omega \rightarrow  \mathbb{R}^3$, 
we consider the functions 
$v^{(h)} : (0,L)\rightarrow \mathbb{R}^3$, $w^{(h)}, u^{(h)} : (0,L) \rightarrow\mathbb{R}$, defined as
\begin{align}
v^{(h)}(s) :=& \,\frac{1}{h^{\alpha - 2}}\int_{D}\big(Y^{(h)}(s,\xi,\zeta) - \Psi^{(h)}(s,\xi,\zeta)\big)\,d\xi\,d\zeta \label{vh},\\
w^{(h)}(s) :=& \,\frac{1}{h^{\alpha - 1}}\left(\frac{1}{\mu(D)}\int_{D}\big(Y^{(h)}(s,\xi,\zeta) - \Psi^{(h)}(s,\xi,\zeta)\big)\cdot (\xi\,\nu_{3}(s) - \zeta\,\nu_{2}(s))\,\,d\xi\,d\zeta\right)\label{wh},\\
u^{(h)}(s) :=& \left\{
\begin{array}{ll}
\vspace{.15cm}
\displaystyle\frac{1}{h^{2(\alpha - 2)}}\int_{s_h}^{s} \Big(\int_{D} \partial_s\big(Y^{(h)}(s,\xi,\zeta) - \Psi^{(h)}(s,\xi,\zeta)\big)\cdot \tau(\sigma) \,\,d\xi\,d\zeta\Big)\,d\sigma  \,\, \mbox{if }\, 2<\alpha<3,\\
\displaystyle\frac{1}{h^{\alpha - 1}}\int_{s_h}^{s} \Big(\int_{D} \partial_s\big(Y^{(h)}(s,\xi,\zeta) - \Psi^{(h)}(s,\xi,\zeta)\big)\cdot \tau(\sigma) \,\,d\xi\,d\zeta\Big)\,d\sigma  \,\quad\, \mbox{if }\, \alpha\geq 3,
\end{array}
\right.\label{uh}
\vspace{.2cm}
\end{align}
where  $s_h \in (0,L)$ is chosen in such a way that $u^{(h)}$ has zero average on $(0,L)$ and 
$\mu(D) := \int_{D}\big(\xi^2 + \zeta^2\big)\,d\xi\,d\zeta$.

Notice that $v^{(h)}$ is the averaged displacement associated with the deformation $Y^{(h)}$. The function $w^{(h)}$ describes the twist of the cross-section. Finally, $u^{(h)}$ is related to the tangential component of the displacement. More precisely, up to a suitable scaling, its derivative $\big(u^{(h)}\big)'$ coincides with the average on the cross-section $D$ of the tangential divergence of $Y^{(h)} - \Psi^{(h)}$. 

We are now in a position to prove the compactness result.
\begin{thm}\label{compactness}
Let $\big(\tilde{y}^{(h)}\big)\subset W^{1,2}\big(\widetilde{\Omega}_{h};\mathbb{R}^{3}\big)$  be a sequence verifying
\begin{equation}\label{finite}
\frac{1}{h^{2\,\alpha - 2}}\,\tilde{I}^{(h)}(\tilde{y}^{(h)})
\leq c < +\infty
\end{equation}
for every $h > 0$. Then there exist an associated sequence $R^{(h)}\subset C^{\infty}((0, L);\mathbb{M}^{3\times 3})$
and constants $\bar{R}^{(h)} \in SO(3)$, $c^{(h)} \in \mathbb{R}^3$ such that, if we define $Y^{(h)}:= \big(\bar{R}^{(h)}\big)^{T}\,\tilde{y}^{(h)}\circ \Psi^{(h)} - c^{(h)}$, we have
\begin{align}
& R^{(h)}(s) \in SO(3) \quad \hbox{for every } s\in (0,L), \label{zero}\\
\vspace{.18cm}
&\big|\big|R^{(h)} - Id\big|\big|_{L^{\infty}(0, L)} \leq C\,h^{\alpha - 2},\quad \big|\big|\big(R^{(h)}\big)'\big|\big|_{L^{2}(0, L)} < C\,h^{\alpha - 2},\label{tre}\\
\vspace{.18cm}
&\big|\big|\nabla_{h} Y^{(h)}\big(\nabla_{h}\Psi^{(h)}\big)^{-1} - R^{(h)}\big|\big|_{L^{2}(\Omega)} \leq C\,h^{\alpha - 1}.\label{uno}
\end{align}
Moreover, defining $v^{(h)}$, $w^{(h)}$ and $u^{(h)}$ as in  (\ref{vh}), (\ref{wh}) and (\ref{uh}), we
 have that, up to subsequences, the following properties are satisfied:
\vspace{.2cm}
\begin{itemize}
\item[(a)] \,$v^{(h)}\rightarrow v$ \, strongly in $W^{1,2}((0, L);\mathbb{R}^{3})$, with
$v\in W^{2,2}((0, L);\mathbb{R}^{3})$ and $v'\cdot \tau = 0$;
\vspace{.15cm}
\item[(b)]\,$w^{(h)}\rightharpoonup w$ \, weakly in $W^{1,2}(0, L)$;
\vspace{.15cm}
\item[(c)]$\left\{
\begin{array}{ll}
\vspace{.15cm}
u^{(h)}\rightarrow u \,\, \hbox{strongly in } W^{1,2}(0, L) \, &\hbox{if } \, 2<\alpha<3, \\
u^{(h)} \,\rightharpoonup u \,\, \hbox{weakly in } W^{1,2}(0, L) \, &\hbox{if } \, \alpha \geq 3.
\end{array}
\right.$
\vspace{.15cm}

In addition, for $2<\alpha<3$ the function $u$ satisfies the following constraint:
\begin{equation}\label{constraint0}
u' = - \frac{1}{2}\Big((v'\cdot \nu_2)^2 + (v'\cdot \nu_3)^2\Big); 
\end{equation}
\vspace{.15cm}
\item[(d)] $\big(\nabla_{h} Y^{(h)}\,\big(\nabla_{h}\Psi^{(h)}\big)^{-1} - Id\,\big)/h^{\alpha - 2} \rightarrow A$
strongly in $L^2(\Omega;\mathbb{M}^{3\times 3})$, where the matrix $A\in W^{1,2}((0, L);\mathbb{M}^{3\times 3})$ is given by
\begin{equation}\label{A}
A = \,R_0 \left(\begin{array}{c}
0\\
v'\cdot\nu_2\\
v'\cdot\nu_3
\end{array}
\begin{array}{c}
- \,v'\cdot\nu_2\\
0\\
w
\end{array}
\begin{array}{c}
-\, v'\cdot\nu_3\\
- \,w\\
0
\end{array} \right)\,R_0^T;
\end{equation}
\vspace{.15cm}
\item[(e)] $\big(R^{(h)} - Id\big)/h^{\alpha - 2} \rightharpoonup A$ weakly in $W^{1,2}((0, L);\mathbb{M}^{3\times 3})$;
\vspace{.15cm}
\item[(f)] $sym\big(R^{(h)} - Id\big)/h^{2(\alpha - 2)} \rightarrow A^2 /2$ \, uniformly on $(0, L)$.
\end{itemize}
\end{thm}
\begin{proof} 
Let $\big(\tilde{y}^{(h)}\big)$ be a sequence in
$W^{1,2}(\widetilde{\Omega}_{h};\mathbb{R}^{3})$ satisfying (\ref{finite});
using the change of variables $\Psi^{(h)}$ and the fact that $\det (\nabla\Psi^{(h)}) = h^2 \det( \nabla_{h}\Psi^{(h)})$,
this estimate becomes
\begin{equation*}
\frac{1}{h^{2\,\alpha - 2}}\int_{\Omega}W\big(s,\xi,\zeta,\nabla\tilde{y}^{(h)}\circ
\Psi^{(h)}\big)\det \big(\nabla_{h}\Psi^{(h)}\big)
ds\,d\xi\,d\zeta \leq c.
\end{equation*}
The coercivity assumption (v) and  (\ref{convdet}) imply that
\begin{equation*}
\frac{1}{h^{2\,\alpha - 2}}\int_{\Omega}\textrm{dist}^{2}\big(\nabla\tilde{y}^{(h)}\circ\Psi^{(h)}, SO(3)\big) ds\,d\xi\,d\zeta \leq c.
\end{equation*}
\indent
\textit{Step 1: Construction of the approximating sequence of rotations.}\\
As in the proof of the compactness result in \cite{MM03} and \cite{S06}, the key tool
is the rigidity theorem \ref{Teorigid}. The idea is
to divide the domain $\widetilde{\Omega}_{h}$ in small
curved cylinders, which are images of homotetic straight cylinders 
through the same Bilipschitz function. Then, we can apply the rigidity 
theorem to each small curved cylinder with the same
constant. In this way we construct a 
sequence of piecewise constant rotations
$Q^{(h)}: [0, L] \rightarrow SO(3)$ satisfying the estimate
\begin{equation}\label{est1}
\int_{\Omega}\big|\nabla\tilde{y}^{(h)}\circ \Psi^{(h)} -Q^{(h)} \big|^2 ds\,d\xi\,d\zeta < c\,h^{2\,\alpha - 2}.
\end{equation}
We include the details for the convenience of the reader. For every small enough $h>0$, let $K_{h}\in
\mathbb{N}$ satisfy $h \leq \frac{L}{K_{h}} < 2\,h$.
For every $a\in[0, L)\cap
\dfrac{L}{K_{h}}\,\mathbb{N}$, define the segments
$$
S_{a,K_{h}}:=
\left\{
\begin{array}{ll}
\vspace{.2cm}
(a, a + 2\,h) & \mbox{if } \,\, a< L-\frac{L}{K_{h}},\\
(L - 2\,h,L ) & \mbox{otherwise}.
\end{array}
\right.
$$
Now consider the cylinders $C_{a,h}:= S_{a,K_{h}}\times D$
and the subsets of $\widetilde{\Omega}_{h}$ defined by
$\widetilde{C}_{a,h}:=
\Psi^{(h)}(C_{a,h})$. Remark that 
$\widetilde{C}_{a,h}$ is a Bilipschitz image of a cube of size $h$, that is 
$(a,0,0) + h\,\big((0, 2)\times D \big)$, through  the map $\Psi$ defined as
\begin{equation*}
\Psi :  [0, L]\times\mathbb{R}^{2} \rightarrow \mathbb{R}^{3}, \quad (s,y_{2},y_{3})\mapsto \gamma(s) + y_{2}\,\nu_{2}(s) + y_{3}\,\nu_{3}(s). 
\end{equation*}
By Theorem \ref{Teorigid} we obtain that there exists a constant rotation $\widetilde{Q}_{a}^{(h)}$ such
that
\begin{equation}\label{rigid}
\int_{\widetilde{C}_{a,h}}\big|\,\nabla\tilde{y}^{(h)}
- \widetilde{Q}_{a}^{(h)}\big|^{2} dx \leq c
\int_{\widetilde{C}_{a,h}}\mbox{dist}^{2}(\nabla\tilde{y}^{(h)},SO(3))
dx.
\end{equation}
In particular, since $\Psi^{(h)}\big(\big(a, a +
\frac{L}{K_{h}}\big)\times D\big)\subset \widetilde{C}_{a,h}$, we
get
\begin{equation}\label{rigidity}
\int_{\Psi^{(h)}\big(\big(a, a +
\frac{L}{K_{h}}\big)\times
D\big)}\big|\,\nabla\tilde{y}^{(h)} -
\widetilde{Q}_{a}^{(h)}\big|^{2} dx \leq c
\int_{\widetilde{C}_{a,h}}\mbox{dist}^{2}(\nabla\tilde{y}^{(h)},SO(3))
dx.
\end{equation}
Changing variables in the integral on the left-hand side, inequality (\ref{rigidity}) becomes
\begin{align*}
\int_{\big(a, a + \frac{L}{K_{h}}\big)\times
D}\big|\,\nabla\tilde{y}^{(h)}\circ\Psi^{(h)} -
\widetilde{Q}_{a}^{(h)}\big|^{2}\det
\big(\nabla\Psi^{(h)}\big) d&s\,d\xi\,d\zeta \leq
c\,\int_{\widetilde{C}_{a,h}}\mbox{dist}^{2}\big(\nabla\tilde{y}^{(h)},SO(3)\big)
dx\\
&\leq
c\,\int_{\widetilde{C}_{a,h}}W\big(\big(\Psi^{(h)}\big)^{-1}
(x),\nabla\tilde{y}^{(h)}(x)\big) dx.
\end{align*}
Notice that $\det \big(\nabla\Psi^{(h)}\big) =
h^2 \det
\big(\nabla_{h}\Psi^{(h)}\big)$ and, since
$\det
\big(\nabla_{h}\Psi^{(h)}\big)\rightarrow 1$
uniformly,
\begin{equation}\label{rigidity2}
\int_{\big(a, a + \frac{L}{K_{h}}\big)\times
D}\big|\,\nabla\tilde{y}^{(h)}\circ\Psi^{(h)} -
\widetilde{Q}_{a}^{(h)}\big|^{2} ds\,d\xi\,d\zeta \leq
\frac{c}{h^2}\,\int_{\widetilde{C}_{a,h}}W\big(\big(\Psi^{(h)}\big)^{-1}
(x),\nabla\tilde{y}^{(h)}(x)\big) dx.
\end{equation}
Now define the map $Q^{(h)}: [0, L)\rightarrow SO(3)$
given by
\begin{equation*}
Q^{(h)}(s):= \widetilde{Q}_{a}^{(h)} \quad
\mbox{for}\, s\in \Big[a, a + \frac{L}{K_{h}}\Big),\,
a\in [0, L)\cap \frac{L}{K_{h}}\,\mathbb{N}.
\end{equation*}
Summing (\ref{rigidity2}) over $a\in [0, L)\cap
\frac{L}{K_{h}}\,\mathbb{N}$ leads to
\begin{equation*}
\int_{\Omega}\big|\,\nabla\tilde{
y}^{(h)}\circ\Psi^{(h)} -
Q^{(h)}\big|^{2}ds\,d\xi\,d\zeta \leq
\frac{c}{h^2}\,\int_{\widetilde{\Omega}_{h}}W\big(\big(\Psi^{(h)}\big)^{-1}
(x),\nabla\tilde{y}^{(h)}(x)\big) dx
\end{equation*}
for a suitable constant independent of $h$. By (\ref{finite}) we
obtain exactly (\ref{est1}).

A similar argument as in the proof of \cite[Theorem 2.2]{MGMM04} shows that 
for every $s\in (h, L-h)$ and every $|\delta| < h$
\begin{equation}\label{puntuale}
\big|Q^{(h)}(s + \delta) - Q^{(h)}(s)\big|^2  \leq c\,h^{2\,\alpha - 3},
\end{equation}
and that for every $I'\subset\subset (0, L)$ and every $\delta\in \mathbb{R}$ with
$|\delta| < \textrm{dist}(I',\{0, L\})$ 
\begin{equation}\label{increm}
\int_{I'}\big|Q^{(h)}(s + \delta) - Q^{(h)}(s)\big|^2 ds \leq c\,h^{2(\alpha - 2)}(|\delta| + h)^2,
\end{equation}
with $c$ independent of $I'$ and $\delta$.
Now, let $\eta\in C_{0}^{\infty}(0,1)$ be such that $\eta\geq 0$, and $\int_{0}^{1}\eta(t)\,dt =1$. We set
$\eta^{(h)}(t):= \frac{1}{h}\eta(\frac{t}{h})$ and we define, as in the proof of \cite[Theorem 2.2]{MGMM04}, 
\begin{equation*}
\tilde{Q}^{(h)}(s):= \int_{0}^{h}\eta^{(h)}(t)Q^{(h)}(s - t)\,dt, \quad s\in[0,L],
\end{equation*}
where we have extended $Q^{(h)}$ out of $(0,L)$ putting $Q^{(h)}(s) := Q^{(h)}(0)$ 
for $s\leq 0$ and $Q^{(h)}(s) := Q^{(h)}(L)$  for $s\geq L$.

By (\ref{puntuale}) and (\ref{increm}) it easily follows that, for every $h>0$,
\begin{align}
&\big|\big|\tilde{Q}^{(h)} - Q^{(h)}\big|\big|_{L^{2}(0, L)} \leq C\,h^{\alpha - 1}, \,\, \big|\big|\big(\tilde{Q}^{(h)}\big)'\big|\big|_{L^{2}(0, L)} \leq c\,h^{\alpha - 2},\label{stimetta}\\
& \big|\big|\tilde{Q}^{(h)} - Q^{(h)}\big|\big|^2_{L^{\infty}(0, L)} \leq C\,h^{2\,\alpha - 3}.\label{jen}
\end{align}

In particular, estimates (\ref{est1}) and (\ref{stimetta}) yield
\begin{equation}\label{314bis}
\big|\big|\nabla\tilde{y}^{(h)}\circ\Psi^{(h)} - \tilde{Q}^{(h)}\big|\big|
_{L^{2}(\Omega)} \leq c\,h^{\alpha - 1}.
\end{equation}
Let $\pi : U \rightarrow SO(3)$ be a smooth projection from a neighborhood $U$ of $SO(3)$ onto $SO(3)$.
From (\ref{jen}) it is clear that the functions $\tilde{Q}^{(h)}$ take values in $U$ for $h$ small enough;
therefore, we can define $\tilde{R}^{(h)}:= \pi\big(\tilde{Q}^{(h)}\big)$. Since 
$\big|\big|\big(\tilde{R}^{(h)}\big)'\big|\big|_{L^2(0,L)} \leq C\,h^{\alpha - 2}$ by (\ref{stimetta}), 
using Sobolev-Poincar\'e inequality we deduce
\begin{equation}\label{nuova}
\big|\big|\tilde{R}^{(h)} - P^{(h)} \big|\big|_{L^{\infty}(0,L)} \leq \big|\big|\big(\tilde{R}^{(h)}\big)'\big|\big|_{L^2(0,L)}\leq c\,h^{\alpha - 2},
\end{equation}
where $P^{(h)}$ is the mean value of $R^{(h)}$ over $(0,L)$.
This implies that 
$$\hbox{dist}\big(P^{(h)},SO(3)\big) \leq c\,h^{\alpha - 2},$$
so there exists a sequence of constant rotations $\big(\bar{R}^{(h)}\big)$ such that 
$\big|\,P^{(h)} - \bar{R}^{(h)} \big| \leq c\,h^{\alpha - 2}$. By this and (\ref{nuova}) we get
\begin{equation*}
\big|\big|\tilde{R}^{(h)} - \bar{R}^{(h)} \big|\big|_{L^{\infty}(0,L)} \leq 
\big|\big|\tilde{R}^{(h)} - P^{(h)} \big|\big|_{L^{\infty}(0,L)} + \big|\,P^{(h)} - \bar{R}^{(h)} \big|
\leq c\,h^{\alpha - 2}.
\end{equation*}
Finally, define $R^{(h)}:= \big(\bar{R}^{(h)}\big)^T \tilde{R}^{(h)}$; this sequence is of 
class $C^{\infty}$ and satisfies (\ref{zero}) and (\ref{tre}). 
Moreover, from (\ref{314bis}) we obtain
\begin{equation}\label{est2}
\big|\big|\nabla\big(\big(\bar{R}^{(h)}\big)^T \tilde{y}^{(h)}\big) \circ 
\Psi^{(h)} -R^{(h)} \big|\big|_{L^2(\Omega)} < C\,h^{\alpha - 1}.
\end{equation}
Let $c^{(h)}\in\mathbb{R}^3$ be the average of the function 
$\big(\bar{R}^{(h)}\big)^{T}\,\tilde{y}^{(h)}\circ \Psi^{(h)} - \Psi^{(h)}$ on
$\Omega$ and let us define the sequence 
$Y^{(h)}:= \big(\bar{R}^{(h)}\big)^{T}\,\tilde{y}^{(h)}\circ \Psi^{(h)} - c^{(h)}$.
Then we can write (\ref{est2}) in terms of $\nabla_{h}Y^{(h)}$ and we get
\begin{equation}\label{est3}
\big|\big|\nabla_{h} Y^{(h)}\big(\nabla_{h}\Psi^{(h)}\big)^{-1} - R^{(h)}\big|\big|_{L^{2}(\Omega)} \leq C\,h^{\alpha - 1},
\end{equation}
which is exactly (\ref{uno}).

\textit{Step 2: Definition of the matrix $A$.}\\ 
As in the case of a straight rod treated in \cite{MGMM04}, we consider the sequence $A^{(h)}$ defined as
$$A^{(h)}(s) := \frac{1}{h^{\alpha - 2}}\big(R^{(h)}(s) - Id\big),$$
which converges uniformly and weakly in $W^{1,2}$ to a matrix $A\in
W^{1,2}((0, L); \mathbb{M}^{3\times 3})$. This is exactly property (e). 
Since $R^{(h)}\in SO(3)$, we have
\begin{equation}\label{simm}
A^{(h)} + \big(A^{(h)}\big)^T = -\, h^{\alpha - 2}\,\big(A^{(h)}\big)^T\,A^{(h)}.
\end{equation}
Passing to the limit as $h \rightarrow 0$, we deduce that $A$ is skew-symmetric. 
Moreover, after division by $2\,h^{\alpha - 2}$ in (\ref{simm}), we get
$$\frac{1}{h^{2(\alpha - 2)}}\,\textrm{sym}\,\big(R^{(h)} - Id\big) \,\rightarrow \frac{A^2}{2}\quad \mbox{uniformly},$$
so property (f) follows.
The convergence of the sequence $A^{(h)}$, together with the estimate (\ref{uno}), imply that
\begin{equation}\label{convA}
\frac{1}{h^{\alpha - 2}}\Big(\nabla_{h} Y^{(h)}\big(\nabla_{h}\Psi^{(h)}\big)^{-1} - Id\Big) \rightarrow A \quad 
\hbox{strongly in }L^{2}(\Omega;\mathbb{M}^{3\times 3}).
\end{equation}
\indent
\textit{Step 3: Identification of $A$ via limiting deformations $v$ and $w$.}\\
Now we characterize the elements of $A$ in terms of some
limiting deformations. By (\ref{convdet}) and (\ref{convA}) we get
\begin{equation}\label{convvA}
\frac{1}{h^{\alpha - 2}}\,\nabla_{h}\big(Y^{(h)} - \Psi^{(h)}\big) \rightarrow A\,R_0 \quad \mbox{strongly in }\, L^{2}(\Omega;\mathbb{M}^{3\times 3}),
\end{equation}
so, in particular, 
\begin{equation}\label{convAfirst}
\frac{1}{h^{\alpha - 2}}\,\partial_s\big(Y^{(h)} - \Psi^{(h)}\big) \rightarrow A\,\tau \quad \mbox{strongly in }\, L^{2}(\Omega;\mathbb{R}^{3}).
\end{equation}
Let $v^{(h)}$ be the sequence introduced in (\ref{vh}). By the choice of $c^{(h)}$,
it has zero average on $(0, L)$ and by (\ref{convAfirst}) its derivative converges strongly 
in $L^{2}((0,L);\mathbb{R}^{3})$ to $A\,\tau$. Therefore, by Poincar\'e
inequality, there exists a function
$v\in W^{1,2}((0,L);\mathbb{R}^{3})$ such that
\begin{equation*}
v^{(h)}\rightarrow v \quad \mbox{strongly in }
\,W^{1,2}((0,L);\mathbb{R}^{3}).
\end{equation*}
Moreover, by (\ref{convAfirst}) we obtain that $v' = A\,\tau$.
As $A$ belongs to $W^{1,2}((0,L);\mathbb{M}^{3\times 3})$ and is skew-symmetric, we deduce that
$v\in W^{2,2}((0,L);\mathbb{R}^{3})$ and $v'\cdot \tau = 0$. Then (a) is proved.

Considering the second and the third columns in (\ref{convvA}) we
have
\begin{equation}\label{convA23}
\frac{1}{h^{\alpha - 1}}\,\partial_\xi \big(Y^{(h)} - \Psi^{(h)}\big) \rightarrow A\,\nu_2 \quad \mbox{and} \quad
\frac{1}{h^{\alpha - 1}}\,\partial_\zeta \big(Y^{(h)} - \Psi^{(h)}\big) \rightarrow A\,\nu_3 \quad
\mbox{strongly in }\, L^{2}(\Omega;\mathbb{R}^{3}).
\end{equation}

\noindent
If we apply Poincar\'e inequality to the function $Y^{(h)} - \Psi^{(h)}$ on $D$, we get 
\begin{equation}\label{Poin}
\big|\big|\,Y^{(h)} - \Psi^{(h)} - \big(Y^{(h)} - \Psi^{(h)}\big)_D\,\big|\big|^2_{L^{2}(D)} \leq c \Big(\big|\big|\,\partial_\xi \big(Y^{(h)} - \Psi^{(h)}\big)\,\big|\big|^2_{L^{2}(D)} + 
\big|\big|\,\partial_\zeta \big(Y^{(h)} - \Psi^{(h)}\big)\,\big|\big|^2_{L^{2}(D)}\Big)
\end{equation}
for a.e. $s\in (0,L)$, where $\big(Y^{(h)} - \Psi^{(h)}\big)_D(s):= \int_D \big( Y^{(h)} - \Psi^{(h)}\big) \,d\xi\,d\zeta$. 
Integrating both sides of (\ref{Poin}) with respect to $s$, we obtain that the sequence 
$\big(Y^{(h)} - \Psi^{(h)} - \big(Y^{(h)} - \Psi^{(h)}\big)_D\big)/h^{\alpha - 1}$ is bounded in $L^2(\Omega;\mathbb{R}^3)$; moreover, (\ref{convA23}) yields that there exists a function $q\in L^2((0,L);\mathbb{R}^3)$ such that
\begin{equation}\label{succ}
\frac{1}{h^{\alpha - 1}}\big( Y^{(h)} - \Psi^{(h)} - \big(Y^{(h)} - \Psi^{(h)}\big)_D\big) \rightarrow \xi\,A\,\nu_2 + \zeta\, A\,\nu_3 + q\quad \mbox{strongly in }\, L^{2}(\Omega;\mathbb{R}^{3}).
\end{equation}
Let $w^{(h)}$ be the sequence defined in (\ref{wh}). Thanks to (\ref{dom2}), it can be rewritten as
\begin{equation}\label{323}
w^{(h)} = \frac{1}{h^{\alpha - 1}}\,\frac{1}{\mu(D)}\int_{D}\left(Y^{(h)} - \Psi^{(h)} - (Y^{(h)} - \Psi^{(h)}\big)_D\right)\cdot (\xi\,\nu_{3} - \zeta\,\nu_{2})\,d\xi\,d\zeta.
\end{equation}
From this expression it is clear that, using (\ref{succ}),
\begin{equation}\label{324}
w^{(h)} \rightarrow w = \frac{1}{\mu(D)}\int_{D}\big(\xi\,A\,\nu_2 + \zeta\, A\,\nu_3\big)\cdot (\xi\,\nu_{3} - \zeta\,\nu_{2})\,d\xi\,d\zeta = (A\,\nu_2)\cdot \nu_3
\end{equation}
strongly in $L^{2}(0, L)$, where the last equality follows from (\ref{dom1}) and from the fact that $A$ is skew-symmetric. It remains to show that the convergence in (\ref{324}) is actually weak in $W^{1,2}(0, L)$. To this aim it is enough to verify the boundedness of the derivative of $w^{(h)}$ in the $L^2$- norm. We get
\begin{align}\label{whs}
\big(w^{(h)}\big)' =&\, \frac{1}{h^{\alpha - 1}}\bigg(\frac{1}{\mu(D)}\int_{D}\partial_s(Y^{(h)} - \Psi^{(h)})\cdot (\xi\,\nu_{3} - \zeta\,\nu_{2})\,d\xi\,d\zeta\bigg) + \nonumber\\
+&\,\frac{1}{h^{\alpha - 1}}\bigg(\frac{1}{\mu(D)}\int_{D}\big(Y^{(h)} - \Psi^{(h)}\big)\cdot (\xi\,\nu'_{3} - \zeta\,\nu'_{2})\,d\xi\,d\zeta\bigg).
\end{align}
For the last integral on the right-hand side of (\ref{whs}) the required bound 
can be proved using the convergence in (\ref{succ}), arguing in a similar way to 
(\ref{323})-(\ref{324}). For the first integral notice that 
\begin{align*}
\frac{1}{h^{\alpha - 1}}\int_{D}\partial_s(Y^{(h)} - \Psi^{(h)})\cdot (\xi\,\nu_{3} - \zeta\,\nu_{2})\,d\xi\,d\zeta=&
\,\frac{1}{h^{\alpha - 1}}\int_{D}\big(\partial_s Y^{(h)} - R^{(h)}\partial_s\Psi^{(h)}\big)\cdot (\xi\,\nu_{3} - \zeta\,\nu_{2})\,d\xi\,d\zeta\\ 
+&\,\frac{1}{h^{\alpha - 1}}\int_{D}\big(R^{(h)}\partial_s\Psi^{(h)} - \partial_s \Psi^{(h)}\big)\cdot (\xi\,\nu_{3} -
\zeta\,\nu_{2})\,d\xi\,d\zeta.
\end{align*}
In virtue of (\ref{est3}) and (\ref{convdet}), the first term on right-hand side is bounded in $L^2$, hence 
it remains to control the $L^{2}$-norm of the second integral. Now, using (\ref{psi}), we have
\begin{equation*}
\int_{D}\big(R^{(h)}\partial_s\Psi^{(h)} - \partial_s \Psi^{(h)}\big)\cdot (\xi\,\nu_{3} -
\zeta\,\nu_{2})\,d\xi\,d\zeta
= h\int_{D}\Big[\big(R^{(h)} - Id\big)(\xi\,\nu_2' +
\zeta\,\nu_3')\Big]\cdot (\xi\,\nu_{3} - \zeta\,\nu_{2})\,d\xi\,d\zeta.
\end{equation*}
The required bound follows from (\ref{tre}), hence (b) is shown.

As $A$ is skew-symmetric and $A\,\tau = v'$, $(A\,\nu_2)\cdot \nu_3 = w$, we conclude that 
\begin{equation*}
R_0^T A\, R_0 = \, \left(\begin{array}{c}
0\\
v'\cdot\nu_2\\
v'\cdot\nu_3
\end{array}
\begin{array}{c}
- \,v'\cdot\nu_2\\
0\\
w
\end{array}
\begin{array}{c}
-\, v'\cdot\nu_3\\
- \,w\\
0
\end{array} \right),
\end{equation*}
which gives (\ref{A}).

\textit{Step 4: Convergence of the sequence $\big(u^{(h)}\big)$.}\\
Let $\big(u^{(h)}\big)$ be the sequence defined in (\ref{uh}).

Consider first the case $2 < \alpha < 3$. It is easy to verify 
that its derivative is bounded in $L^{2}(0, L)$. Indeed,
\begin{align*}
\big(u^{(h)}\big)' =&\, \frac{1}{h^{2(\alpha - 2)}}\int_{D}\partial_{s}\big(Y^{(h)} - 
\Psi^{(h)}\big)\cdot\tau \,d\xi\,d\zeta = \,\frac{1}{h^{2(\alpha - 2)}}\int_{D}\big(\partial_{s}Y^{(h)} - R^{(h)}\partial_{s}\Psi^{(h)}\big)\cdot\tau \,d\xi\,d\zeta \nonumber\\ 
&\,+\,  \frac{1}{h^{2(\alpha - 2)}}\int_{D}\big(R^{(h)}\partial_{s}\Psi^{(h)} - \partial_{s}\Psi^{(h)}\big)\cdot\tau\,d\xi\,d\zeta.
\end{align*}
Since $\alpha < 3$, the first term converges to zero strongly in $L^{2}$ by 
(\ref{convdet}) and (\ref{est3}). As for the second term, using (\ref{psi}), the
fact that $R^{(h)}$ is independent of $\xi$ and $\zeta$ and (\ref{dom2}), we have
\begin{align*}
\frac{1}{h^{2(\alpha - 2)}}\int_{D}\big(R^{(h)}\partial_{s}\Psi^{(h)} - 
\partial_{s}\Psi^{(h)}\big)\cdot\tau\,d\xi\,d\zeta &= \,
\frac{1}{h^{2(\alpha - 2)}}\,\big(R^{(h)}\tau - \tau\big)\cdot\tau\\
&=\, \frac{1}{h^{2(\alpha - 2)}}\,\hbox{sym}\big(R^{(h)} - Id\big)\tau\cdot\tau.
\end{align*}
By property (f) this converges to $(A^2\tau)\cdot\tau/2$ uniformly on $(0,L)$.
As $u^{(h)}$ has zero average, by Poincar\'e inequality we deduce that $u^{(h)}$ 
converges to $u$ strongly in $W^{1,2}$, where $u$ satisfies
\begin{equation}\label{constraint}
u' = \left(\frac{A^2}{2}\, \tau\right)\cdot \tau = - \frac{1}{2}\Big((v'\cdot \nu_2)^2 + (v'\cdot \nu_3)^2\Big). 
\end{equation}
In the case $\alpha\geq 3$ the derivative of $\big(u^{(h)}\big)$ can be written as
\begin{align*}
\big(u^{(h)}\big)' = \,\frac{1}{h^{\alpha - 1}}\int_{D}\big(\partial_{s}Y^{(h)} - R^{(h)}\partial_{s}\Psi^{(h)}\big)
\cdot\tau \,d\xi\,d\zeta + \frac{1}{h^{\alpha - 1}}\,\hbox{sym}\big(R^{(h)} - Id\big)\tau\cdot\tau.
\end{align*}
The first term is bounded in $L^{2}(0, L)$ by (\ref{convdet}) and (\ref{est3}), while the second term converges 
to zero uniformly by (f). 

This concludes the proof of (c) and of the theorem.
\end{proof}

\section{Liminf inequalities} 
\noindent
In this section we will show a lower bound for the energy $\big(\tilde{I}^{(h)}\big)/h^{(2\,\alpha -2)}$, 
for all the scalings $\alpha > 2$, and we will describe the limiting functionals.

Let $Q_{3}: \Omega\times\mathbb{M}^{3\times 3}\longrightarrow [0, +\infty)$ be
twice the quadratic form of linearized elasticity, i.e.,
\begin{equation}\label{defQ3}
Q_{3}(z,G) := \frac{\partial^2 W}{\partial F^2}(z,Id)[G,G]
\end{equation}
for every $z\in \Omega$ and every $G\in \mathbb{M}^{3\times 3}$. 
Let $Q^{0}: (0, L)\times \mathbb{R}\times \mathbb{M}^{3\times
3}_{skew} \rightarrow [0, +\infty)$ and $Q:(0,L)\times M^{3\times 3}
\rightarrow [0, +\infty)$ \, be defined as
\begin{equation}\label{defQ0}
Q^{0}(s,t,F):= \min_{\varphi\in W^{1,2}(D;\mathbb{R}^3)} \int_{D}
Q_{3}\bigg(s,\xi,\zeta,
R_{0}\bigg(F\Bigg(\begin{array}{c}
0\\
\xi\\
\zeta
\end{array}\Bigg) + t\,e_1 \,\Big|\,\partial_{\xi}\varphi\,\Big|\,\partial_{\zeta}\varphi
\bigg)\,R_{0}^{T}\bigg)\,d\xi\,d\zeta
\end{equation}
and 
\begin{equation}\label{defQ}
Q(s,F):= \min_{t\in\mathbb{R}}\, Q^{0}(s,t,F),
\end{equation}
respectively.
For $u, w \in W^{1,2}(0, L)$ and $v\in W^{2,2}((0,L);\mathbb{R}^3)$ we introduce the functionals
\begin{equation}\label{defI0}
I_{\alpha}(u,v,w):= \left\{
\begin{array}{ll}
\vspace{.15cm}
\displaystyle \frac{1}{2}\int_{0}^{L}Q^0\Big(s, u' + \frac{1}{2}\big((v'\cdot\nu_2)^2 + (v'\cdot \nu_3)^2\big),
B' + 2\,\hbox{skw}\big(R_0^T R_0' B\big)\Big)\,ds \,\, &\hbox{if } \alpha = 3,\\
\displaystyle \frac{1}{2}\int_{0}^{L}Q^0\big(s, u', B' + 2\,\hbox{skw}\big(R_0^T R_0' B\big)\big)\,ds \qquad &\hbox{if } \alpha > 3,
\end{array}
\right.
\end{equation}
and, for $2<\alpha<3$,
\begin{equation}\label{defI}
I_{\alpha}(v,w):= \frac{1}{2}\int_{0}^{L}Q\big(s, B' + 2\,\hbox{skw}\big(R_0^T R_0' B\big)\big)\,ds,
\end{equation}
where $B\in W^{1,2}((0,L);\mathbb{M}^{3\times 3})$ denotes the
matrix
\begin{equation}\label{An}
B := \left(\begin{array}{c}
0\\
v'\cdot\nu_2\\
v'\cdot\nu_3
\end{array}
\begin{array}{c}
- v'\cdot\nu_2\\
0\\
w
\end{array}
\begin{array}{c}
- v'\cdot\nu_3\\
- w\\
0
\end{array} \right).
\end{equation}
\begin{rem}\label{rk}
It is easy to see that the minimum in (\ref{defQ0}) is attained, it is
unique and it can be computed on the subspace
\begin{equation}\label{defV}
\mathcal{V}:= \Big\{\varphi \in W^{1,2}(D;\mathbb{R}^3) : \,
\int_{D}\varphi\,d\xi\,d\zeta = 0, \,\int_{D}\varphi\cdot (\zeta\,\nu_2 -
\xi\,\nu_3)\,d\xi\,d\zeta = 0\Big\},
\end{equation}
(see \cite[Remark 4.1]{MGMM04}). Moreover the minimizer $\varphi$ depends linearly on the data $t$ and $F$.
More precisely, if $t \in L^2(0,L)$ and $F\in L^2((0,L);\mathbb{M}_{skew}^{3\times 3})$, then denoting with  
$\varphi(s,\cdot)\in \mathcal{V}$ the solution of the problem (\ref{defQ0}) with data $t(s)$ and $F(s)$, we have 
that $\varphi \in L^2(\Omega;\mathbb{R}^3)$ and also $\partial_{\xi}\varphi, \partial_{\zeta}\varphi \in L^2(\Omega;\mathbb{R}^3)$.
Analogously, if $t$ solves (\ref{defQ}), then $t$ depends linearly on $F$. So, if 
$F\in L^2((0,L);\mathbb{M}_{skew}^{3\times 3})$ and $t$ is the solution to (\ref{defQ}) corresponding to $F(s)$, 
then $t\in L^2(0,L)$.
\end{rem}

\begin{rem}\label{comparison}
The limit functionals corresponding to the scalings $2<\alpha<3$ and $\alpha > 3$ turn out to be linear. 
Notice that, in the case $2<\alpha<3$, the deformation $u$ is completely determined by 
$v$ in virtue of the constraint (\ref{constraint0}) in Theorem \ref{compactness}. 
This explains the reason why the $\Gamma$-limit obtained for this scaling does not depend on $u$. 
On the other hand, for $\alpha > 3$, the function $u$ is independent of $v$ and $w$ and the functional 
$I_{\alpha}$ describing the one-dimensional problem coincides with the one obtained by dimension reduction, 
starting from 3D linearized elasticity (see \cite{Griso}, \cite{JuTam} and \cite{Sepp1}).

More precisely, if we assume in addition that the density $W$ is homogeneous and isotropic, that is, 
\begin{equation*}
W(F) = W(FR) \quad \hbox{for every } R\in SO(3),
\end{equation*}
then the quadratic form $Q_3$ is given by 
\begin{equation*}
Q_{3}(G) = 2\,\mu\,\bigg|\frac{G + G^{T}}{2}\bigg|^{2} + \lambda\,(\mbox{tr}\, G)^{2}
\end{equation*}
for some constants $\lambda,\mu \,\in \mathbb{R}$. Since for all 
$G \in\mathbb{M}^{3\times 3}$ and $R\in SO(3)$ we have
\begin{equation*}
Q_{3}(R\,G\,R^{T}) = Q_{3}(G),
\end{equation*}
by \cite[Remark 3.5]{MM03} formula (\ref{defQ0}) reduces to 
\begin{equation}\label{newQ0}
Q^{0}(s,t,F) =\,\frac{\mu(3\,\lambda + 2\,\mu)}{\lambda + \mu}\,(t^2 + I_3 F_{12}^2 + I_2 F_{13}^2) + \mu\,T\,F_{23}^2,
\end{equation}
where $I_3 = \int_D \xi^2 d\xi\,d\zeta$,  $I_2 = \int_D \zeta^2 d\xi\,d\zeta$ and $T$ is the so-called torsional rigidity, which depends on the section. Therefore by (\ref{newQ0}), (\ref{dom1}) and (\ref{dom2}) the limit functional reads as follows
\begin{equation*}
I_{\alpha}(u,v,w) = \frac{1}{2}\,\frac{\mu(3\,\lambda + 2\,\mu)}{\lambda + \mu}\int_{0}^{L} 
\big((u')^2 + I_2 q_2^2 + I_3 q_3^2 \big)\,ds + \frac{1}{2}\,\mu\,T \int_{0}^{L} q_1^2 ds,
\end{equation*}
where
\begin{align*}
q_1 := & \, w' + k_2 (v\cdot \nu_3)' - k_3 (v\cdot \nu_2)' + \varrho\,\big(k_2 (v\cdot \nu_2) + 
k_3 (v\cdot \nu_3)\big),\\
q_2 := & \, k_2 w - (v\cdot \nu_3)'' - 2\,\varrho\,(v\cdot \nu_2)' - 
(v\cdot \tau)\,\big(\varrho\,k_2 + k_3'\big) 
+\,(v\cdot \nu_2)\,\big(k_2 k_3 + \varrho'\big) + (v\cdot \nu_3)\,\big(\varrho^2 - k_3^2\big),\\
q_3 := & \, k_3\,w + (v\cdot \nu_2)'' - 2\,\varrho\,(v\cdot \nu_3)' - 
(v\cdot \tau)\,\big(\varrho\,k_3 - k_2'\big) 
-\,(v\cdot \nu_2)\,\big(\varrho^2 - k_2^2\big) + 
(v\cdot \nu_3)\,\big(k_2\,k_3 - \varrho'\big).
\end{align*}

This is the functional derived in \cite{Griso}, \cite{JuTam} and \cite{Sepp1}, starting from linearized elasticity.
\end{rem}

Now we are ready to show a lower bound for the functionals $h^{- \alpha} \tilde{I}^h$ with $2 < \alpha < 3$.

\begin{thm}[Case $2 < \alpha < 3$]\label{scithm}
Let $w\in W^{1,2}(0,L)$ and let $v\in
W^{2,2}((0,L);\mathbb{R}^3)$ be such that $v'\cdot \tau = 0$. 
Then, for every positive sequence $(h_j)$ converging to zero and every 
sequence $\big(\tilde{y}^{(h_j)}\big)\subset
W^{1,2}(\widetilde{\Omega}_{h_j};\mathbb{R}^{3})$ such that the
sequence $Y^{(h_j)}:= \tilde{y}^{(h_j)}\circ \Psi^{(h_j)}$
satisfies the properties (a), (b) and (d) of Theorem \ref{compactness}, it
turns out that
\begin{equation}\label{fine}
\liminf_{j\rightarrow
\infty}\frac{1}{h_j^{2\,\alpha}}\int_{\widetilde{\Omega}_{h_j}}
W\Big(\big(\Psi^{(h_j)}\big)^{-1} (x),\nabla
\tilde{y}^{(h_j)}(x)\Big) dx \geq 
I_{\alpha}(v, w),
\end{equation}
where $I_{\alpha}$ is introduced in (\ref{defI}).
\end{thm}
\begin{proof}
In the following, we will write simply $h$ instead of $h_j$.
Let $\big(\tilde{y}^{(h)}\big)$ be a sequence such that $Y^{(h)}:=
\tilde{y}^{(h)}\circ \Psi^{(h)}$ satisfies the required
assumptions. 

\textit{First step: lower bound for the energy. }
We can suppose that
\begin{equation*}
\liminf_{h\rightarrow 0}\frac{1}{h^{2\,\alpha}}\int_{\widetilde{\Omega}_{h}}
W\Big(\big(\Psi^{(h)}\big)^{-1} (x),\nabla \tilde{y}^{(h)}(x)\Big)
dx \leq c < +\,\infty,
\end{equation*}
otherwise (\ref{fine}) is trivial. Therefore, up to subsequences,
(\ref{finite}) is satisfied. By Theorem \ref{compactness} we get
the existence of a sequence $R^{(h)}: [0,L]\rightarrow SO(3)$ such
that
\begin{equation}\label{esti2}
||\nabla_h Y^{(h)}\big(\nabla_h \Psi^{(h)}\big)^{-1} -
R^{(h)}||_{L^2(\Omega)} \leq c\,h^{\alpha - 1}
\end{equation}
and $R^{(h)} \rightarrow Id$ uniformly. Define the functions
$G^{(h)}: \Omega\rightarrow \mathbb{M}^{3\times 3}$\, as
\begin{equation}\label{newGh}
G^{(h)}:= \frac{1}{h^{\alpha - 1}}\Big((R^{(h)})^{T}\nabla_h
Y^{(h)}\big(\nabla_h\Psi^{(h)}\big)^{-1} - Id\Big).
\end{equation}
By (\ref{esti2}) they are bounded in
$L^{2}(\Omega;\mathbb{M}^{3\times 3})$, so there exists $G\in
L^{2}(\Omega;\mathbb{M}^{3\times 3})$ such that
$G^{(h)}\rightharpoonup G$ weakly in
$L^{2}(\Omega;\mathbb{M}^{3\times 3})$. We claim that
\begin{equation}\label{liminf}
\liminf_{h\rightarrow 0}\frac{1}{h^{2\,\alpha}}\int_{\widetilde{\Omega}_{h}} W\Big(\big(\Psi^{(h)}\big)^{-1} (x),\nabla \tilde{y}^{(h)}(x)\Big) dx \geq \frac{1}{2} \int_{\Omega} Q_{3}(s,\xi,\zeta,G) ds\,d\xi\,d\zeta.
\end{equation}
Performing the change of variables $\Psi^{(h)}$ and using the
frame indifference of $W$, we have
\begin{align}\label{newvar}
&\frac{1}{h^{2\,\alpha}}\int_{\widetilde{\Omega}_{h}} W\Big(\big(\Psi^{(h)}\big)^{-1} (x),\nabla \tilde{y}^{(h)}\Big) dx  = \frac{1}{h^{2\,\alpha - 2}}\int_{\Omega} W\big(s,\xi,\zeta, \nabla\tilde{y}^{(h)}\circ \Psi^{(h)}\big)\det \big(\nabla_{h}\Psi^{(h)}\big) ds\,d\xi\,d\zeta \nonumber\\
&\hspace{3cm} = \frac{1}{h^{2\,\alpha - 2}}\int_{\Omega} W\Big(s,\xi,\zeta, \big(\nabla_h Y^{(h)}\big)\big(\nabla_{h}\Psi^{(h)}\big)^{-1}\Big)\det \big(\nabla_{h}\Psi^{(h)}\big) ds\,d\xi\,d\zeta.
\end{align}
We introduce the functions
$$
\chi^{(h)}(s,\xi,\zeta):=
\left\{
\begin{array}{ll}
\vspace{.1cm}
\displaystyle 1 & \mbox{if } \,\, \snorm{G^{(h)}(s,\xi,\zeta)} \leq h^{2 - \alpha},\\
\displaystyle 0 & \mbox{otherwise}.
\end{array}
\right.
$$
From the boundedness of $G^{(h)}$ in $L^2(\Omega;\mathbb{M}^{3\times 3})$ we get that 
$\chi^{(h)}\rightarrow 1$ boundedly in measure, so that
\begin{equation}\label{convchi}
\chi^{(h)}G^{(h)} \rightharpoonup G \quad \mbox{weakly in}\, L^{2}(\Omega;\mathbb{M}^{3\times 3}).
\end{equation}
By expanding $W$ around the identity, we obtain that for every $(s,\xi,\zeta) \in \Omega$ and
$A\in \mathbb{M}^{3\times 3}$
\begin{equation*}
W\big(s,\xi,\zeta, Id + A) = \frac{1}{2}\,\frac{\partial^{2}W}{\partial F^{2}}\,(s,\xi,\zeta, Id + t\,A)[A,A],
\end{equation*}
where \,$0<t<1$\, depends on the point $(s,\xi,\zeta)$ and on $A$. By (\ref{newvar}) and by the definition of $G^{(h)}$
we have
\begin{align*}
\frac{1}{h^{2\,\alpha}}\,\tilde{I}^{(h)}\big(\tilde{y}^{(h)}\big) &=
\frac{1}{h^{2\,\alpha - 2}}\int_{\Omega} W\big(s,\xi,\zeta, Id + h^{\alpha - 1}\,G^{(h)} \big)\det \big(\nabla_{h}\Psi^{(h)}\big) ds\,d\xi\,d\zeta \nonumber\\
&\geq \frac{1}{h^{2\,\alpha - 2}}\int_{\Omega} \chi^{(h)} W\big(s,\xi,\zeta,Id + h^{\alpha - 1}\,G^{(h)}\big)\det \big(\nabla_{h}\Psi^{(h)}\big) ds\,d\xi\,d\zeta \nonumber\\
&= \frac{1}{2}\int_{\Omega} \chi^{(h)}\left(\frac{\partial^{2}W}{\partial F^{2}}\,\big(s,\xi,\zeta, Id + h^{\alpha - 1}\,t(h)\,G^{(h)}\big)\big[G^{(h)},G^{(h)}\big]\right)\det \big(\nabla_{h}\Psi^{(h)}\big) ds\,d\xi\,d\zeta,
\end{align*}
where $0 < t(h) < 1$\, depends on $(s,\xi,\zeta)$ and on $G^{(h)}$.
For the last integral in the previous formula we have that
\begin{align}
&\int_{\Omega} \chi^{(h)}\left(\frac{\partial^{2}W}{\partial F^{2}}\,\big(s,\xi,\zeta, Id + h^{\alpha - 1}\,t(h)\,G^{(h)}\big)\big[G^{(h)},G^{(h)}\big]\right)\det \big(\nabla_{h}\Psi^{(h)}\big) ds\,d\xi\,d\zeta = \nonumber\\
&\int_{\Omega}\chi^{(h)}\bigg(\frac{\partial^{2}W}{\partial F^{2}}\,\big(s,\xi,\zeta, Id + h^{\alpha - 1}\,t(h)\,G^{(h)}\big)\big[G^{(h)},G^{(h)}\big] - Q_{3}\big(s,\xi,\zeta, G^{(h)}\big)\Big)\det \big(\nabla_{h}\Psi^{(h)}\big) ds\,d\xi\,d\zeta \nonumber\\
&\hspace{1cm}+ \int_{\Omega}Q_{3}\big(s,\xi,\zeta,\chi^{(h)}\,G^{(h)}\big)\det \big(\nabla_{h}\Psi^{(h)}\big) ds\,d\xi\,d\zeta \label{3.22quater}.
\end{align}

Notice that the second integral is lower semicontinuous with respect to the weak topology of $L^2$; so, the claim follows from (\ref{convchi}), once we prove that the first term in (\ref{3.22quater}) can be neglected for $h$ small enough. To this aim, 
we apply Scorza-Dragoni theorem to the function $\partial^{2}W/\partial F^{2}$ and we have that there exists a compact set $K\subset \Omega$ such that the function \,$\partial^{2}W/\partial F^{2}$ \, restricted to $K\times \overline{B_{\delta}(Id)}$ is continuous, hence uniformly continuous, and the measure of $\Omega\setminus K$ is small.
Since $h^{\alpha - 1}\,t(h)\,\chi^{(h)}\,G^{(h)}$ is uniformly small for $h$ small enough, for every $\varepsilon > 0$ we have
\begin{align*}
\frac{1}{2}\int_{\Omega}&\chi^{(h)}\Bigg(\frac{\partial^{2}W}{\partial F^{2}}\,\big(s,\xi,\zeta, Id + h^{\alpha - 1}\,t(h)\,G^{(h)}\big)\big[G^{(h)},G^{(h)}\big] - Q_{3}\big(s,\xi,\zeta,G^{(h)}\big)\Bigg)\det \big(\nabla_{h}\Psi^{(h)}\big) ds\,d\xi\,d\zeta \\
&\geq - \frac{\varepsilon}{2}\int_{K}\chi^{(h)}\big|\,G^{(h)}\big|^{2}\det \big(\nabla_{h}\Psi^{(h)}\big) ds\,d\xi\,d\zeta
\geq -\,C\,\varepsilon
\end{align*}
for $h$ small enough. Hence, being $\varepsilon$ arbitrary, (\ref{liminf}) is proved.

Since, by frame indifference, the quadratic form $Q_3$ depends only on the 
symmetric part of $G$, we obtain the bound
\begin{equation}\label{sc}
\liminf_{h\rightarrow 0}\frac{1}{h^{2\,\alpha}}\int_{\widetilde{\Omega}_{h}}
W\Big(\big(\Psi^{(h)}\big)^{-1} (x),\nabla \tilde{y}^{(h)}(x)\Big)
dx \geq \frac{1}{2} \int_{\Omega} Q_{3}(s,\xi,\zeta,\tilde{G})
ds\,d\xi\,d\zeta,
\end{equation}
where $\tilde{G}$ denotes the symmetric part of $G$. 

\textit{Second step: identification of $\tilde{G}$.} In order to identify $\tilde{G}$, we first notice that, since $R^{(h)}\rightarrow Id$ uniformly,
\begin{equation*}
R^{(h)}G^{(h)} = \, \frac{1}{h^{\alpha - 1}}\Big(\nabla_h Y^{(h)}\big(\nabla_h\Psi^{(h)}\big)^{-1} - R^{(h)}\Big)\rightharpoonup G
\end{equation*}
weakly in $L^{2}(\Omega;\mathbb{M}^{3\times 3})$; moreover, by (\ref{convdet}),
\begin{equation}\label{GG}
R^{(h)}G^{(h)}\nabla_h\Psi^{(h)} = \, \frac{1}{h^{\alpha - 1}}\Big(\nabla_h Y^{(h)} - R^{(h)}\nabla_h\Psi^{(h)}\Big)
\rightharpoonup G\,R_{0}.
\end{equation}
In particular, considering the second and the third columns in (\ref{GG}) we get
\begin{equation*}
R^{(h)}G^{(h)}\nu_2 = \frac{1}{h^\alpha}\,\partial_{\xi}\big(Y^{(h)} - R^{(h)}\Psi^{(h)}\big) = \frac{1}{h^\alpha}\,\Big(\partial_{\xi}Y^{(h)} - h\,R^{(h)}\nu_2\Big)\rightharpoonup G\,\nu_2
\end{equation*}
and
\begin{equation*}
R^{(h)}G^{(h)}\nu_3 = \frac{1}{h^\alpha}\,\partial_{\zeta}\big(Y^{(h)} - R^{(h)}\Psi^{(h)}\big) = \frac{1}{h^\alpha}\,\Big(\partial_{\zeta}Y^{(h)} - h\,R^{(h)}\nu_3\Big)\rightharpoonup G\,\nu_3.
\end{equation*}
Let us define the functions $\tilde{\beta}^{(h)}: \Omega \rightarrow \mathbb{R}^3$ as
\begin{equation*}
\tilde{\beta}^{(h)}(s,\xi,\zeta) := \frac{1}{h^\alpha}\,\Big(Y^{(h)} - h\,\xi\,R^{(h)}\nu_2 - h\,\zeta\,R^{(h)}\nu_3\Big).
\end{equation*}
Easy computations show that
\begin{equation}\label{colonne23}
\partial_{\xi}\tilde{\beta}^{(h)} = R^{(h)}G^{(h)}\nu_2 \quad\hbox{and} \quad 
\partial_{\zeta}\tilde{\beta}^{(h)} = R^{(h)}G^{(h)}\nu_3,
\end{equation}
hence $\partial_{\xi}\tilde{\beta}^{(h)}$ and $\partial_{\zeta}\tilde{\beta}^{(h)}$ 
are bounded in $L^2(\Omega)$. By Poincar\'e inequality, this implies that 
\begin{equation*}
\big|\big|\tilde{\beta}^{(h)} - \tilde{\beta}_D^{(h)} \big|\big|^2_{L^{2}(\Omega)}
\leq C\,\Big(\big|\big|\partial_{\xi}\tilde{\beta}^{(h)}\big|\big|^2_{L^{2}(\Omega)} + \big|\big|\partial_{\zeta}\tilde{\beta}^{(h)}\big|\big|^2_{L^{2}(\Omega)}\Big) \leq c,
\end{equation*}
where $\tilde{\beta}_D^{(h)}(s):= \int_D \tilde{\beta}^{(h)}(s,\xi,\zeta)\,d\xi\,d\zeta$. Therefore, 
there exists a function $\beta\in L^{2}(\Omega;\mathbb{R}^{3})$ such that
\begin{equation}\label{betaconv}
\beta^{(h)}:=\tilde{\beta}^{(h)} - \tilde{\beta}_D^{(h)} \rightharpoonup \beta \,\,\hbox{weakly in }L^{2}(\Omega;\mathbb{R}^{3}).
\end{equation}
From (\ref{colonne23}), as $h\rightarrow 0$, we get
\begin{equation}\label{GG23}
G\,\nu_2 = \partial_{\xi}\beta \quad\hbox{and} \quad
G\,\nu_3 = \partial_{\zeta}\beta.
\end{equation}
Considering the first columns in (\ref{GG}) we have 
\begin{equation}\label{column1}
R^{(h)}G^{(h)}\partial_s\Psi^{(h)} = \, \frac{1}{h^{\alpha - 1}}\Big(\partial_s Y^{(h)} - 
R^{(h)}\partial_s\Psi^{(h)}\Big) \rightharpoonup G\,\tau.
\end{equation}
Using (\ref{psi}) and the definitions of $\tilde{\beta}_D^{(h)}$ and $\beta^{(h)}$, we can write
\begin{align}\label{colonna1}
R^{(h)}G^{(h)}\partial_s\Psi^{(h)} =&\,
\frac{1}{h^{\alpha - 1}}\Big(\partial_s Y^{(h)} - h\,\xi R^{(h)}\nu_2' - h\,\zeta R^{(h)}\nu_3'\Big) - 
\frac{1}{h^{\alpha - 1}}\,R^{(h)}\tau\nonumber\\
=&\, h\,\partial_s \beta^{(h)} + \frac{1}{h^{\alpha - 2}}\,\big(R^{(h)}\big)'(\xi\,\nu_2 + \zeta\,\nu_3) +
\frac{1}{h^{\alpha - 1}}\int_{D}\big(\partial_s Y^{(h)} - R^{(h)}\tau\big)\,d\xi\,d\zeta.
\end{align}
By (\ref{betaconv}) it follows that
\begin{equation}\label{stellina}
h\,\partial_s \beta^{(h)} \rightharpoonup 0 \quad \hbox{weakly in } \,W^{-1,2}(\Omega;\mathbb{R}^3).
\end{equation}
Moreover, from (\ref{esti2}), it turns out that there exists $g\in L^2((0,L);\mathbb{R}^3)$ such that
\begin{equation}\label{identig}
\frac{1}{h^{\alpha - 1}}\int_{D}\big(\partial_s Y^{(h)} - R^{(h)}\tau\big)\,d\xi\,d\zeta 
= \,\frac{1}{h^{\alpha - 1}}\int_{D}\big(\partial_s Y^{(h)} - R^{(h)}\partial_s \Psi^{(h)}\big)\,d\xi\,d\zeta  \rightharpoonup g
\end{equation}
weakly in $L^2((0,L);\mathbb{R}^3)$. Passing to the limit in (\ref{colonna1}) and
using (\ref{column1}), (\ref{identig}), (\ref{stellina}), and property (e) of 
Theorem \ref{compactness}, we obtain
\begin{equation}\label{bis}
G\,\tau = A'(\xi\,\nu_2 + \zeta\,\nu_3) + g.
\end{equation}
Finally, by (\ref{GG23}) and (\ref{bis}) we have that 
\begin{equation*}
G\,R_0 = \Bigg(A'\,R_0\Bigg(\begin{array}{c}
0\\
\xi\\
\zeta
\end{array}\Bigg) + g\bigg|\,
\partial_{\xi} \beta\,\bigg|\,
\partial_{\zeta} \beta
 \Bigg).
\end{equation*}
As $\hbox{sym}\,\big(R_0^T G\,R_0\big) 
= R_0^T \tilde{G} \,R_0$, we deduce that
\begin{equation*}
R_0^T \tilde{G}\,R_0 = \hbox{sym}\Bigg(R_0^T A'\,R_0\Bigg(\begin{array}{c}
0\\
\xi\\
\zeta
\end{array}\Bigg) + \hat{g} \bigg|\,
\partial_{\xi} \hat{\beta}\,\bigg|\,
\partial_{\zeta} \hat{\beta}
 \Bigg), 
\end{equation*}
where $\hat{\beta}:= R_0^T\beta$ and $\hat{g}:= R_0^T g$.
If we define $\varphi:= \hat{\beta} + \xi\,(\hat{g}\cdot e_2)\,e_1 + 
\zeta\,(\hat{g}\cdot e_3)\,e_1$, we obtain the expression
\begin{equation}\label{vif}
R_0^T  \tilde{G}\,R_0 = \hbox{sym}\Bigg(R_0^T A'\,R_0\Bigg(\begin{array}{c}
0\\
\xi\\
\zeta
\end{array}\Bigg) + (\hat{g}\cdot e_1)\,e_1 \bigg|\,
\partial_{\xi} \varphi\,\bigg|\,
\partial_{\zeta} \varphi
 \Bigg). 
\end{equation}

Now, let us rewrite the previous expression in terms of the matrix $B$ defined in
(\ref{An}), noticing that $A = R_0 B R_0^T$.
It turns out that
\begin{equation*}
A' =  R'_0 B R_0^T +  R_0 B' R_0^T + R_0 B \big(R_0^T\big)',
\end{equation*}
hence
\begin{equation*}
R_0^T A'\,R_0 = R_0^T R'_0 B +  B' + B \big(R_0^T\big)'R_0.
\end{equation*}
Since $B$ is skew-symmetric, we deduce
\begin{equation}\label{vifor}
R_0^T A'\,R_0 = B' + 2\,\hbox{skw}\big(R_0^T R'_0 B\big).
\end{equation}
Using this identity in (\ref{vif}) we have
\begin{equation*}
R_0^T  \tilde{G}\,R_0 = \hbox{sym}\Bigg(\Big(B' + 2\,\hbox{skw}\big(R_0^T R'_0 B\big)\Big)\Bigg(\begin{array}{c}
0\\
\xi\\
\zeta
\end{array}\Bigg) + (\hat{g}\cdot e_1)\,e_1 \bigg|\,
\partial_{\xi} \varphi\,\bigg|\,
\partial_{\zeta} \varphi
 \Bigg). 
\end{equation*}
Finally, as
\begin{equation*}
\hat{g}\cdot e_1 = \big(R_0^T g\big)\cdot e_1 = g\cdot (R_0 e_1) = g\cdot \tau,
\end{equation*}
we conclude that
\begin{equation}\label{vif3}
R_0^T  \tilde{G}\,R_0 = \hbox{sym}\Bigg(\Big(B' + 2\,\hbox{skw}\big(R_0^T R'_0 B\big)\Big)\Bigg(\begin{array}{c}
0\\
\xi\\
\zeta
\end{array}\Bigg) + (g\cdot \tau)\,e_1 \bigg|\,
\partial_{\xi} \varphi\,\bigg|\,
\partial_{\zeta} \varphi
 \Bigg). 
\end{equation}

\textit{Third step: description of the limit functional. }
Since $\varphi(s,\cdot)\in W^{1,2}(D;\mathbb{R}^{3})$ for a.e. $s\in (0, L)$, using (\ref{sc}), (\ref{vif3}) and the definition of $Q$, we obtain exactly (\ref{fine}). 
\end{proof}

It remains to show the lower bound for the functionals $h^{- \alpha} \tilde{I}^h$ with $\alpha \geq 3$.

\begin{thm}[Case $\alpha\geq 3$]\label{scithm2}
Let $u,w\in W^{1,2}(0,L)$ and let $v\in
W^{2,2}((0,L);\mathbb{R}^3)$ be such that $v'\cdot \tau = 0$. 
Then, for every positive sequence $(h_j)$ converging to zero and
every  sequence $\big(\tilde{y}^{(h_j)}\big)\subset
W^{1,2}(\widetilde{\Omega}_{h_j};\mathbb{R}^{3})$ such that the
sequence $Y^{(h_j)}:= \tilde{y}^{(h_j)}\circ \Psi^{(h_j)}$
satisfies the properties (a)-(d) of Theorem \ref{compactness}, it
turns out that
\begin{equation}\label{fine2}
\liminf_{j\rightarrow
\infty}\frac{1}{h_j^{2\,\alpha}}\int_{\widetilde{\Omega}_{h_j}}
W\Big(\big(\Psi^{(h_j)}\big)^{-1} (x),\nabla
\tilde{y}^{(h_j)}(x)\Big) dx \geq 
I_{\alpha}(u, v, w),
\end{equation}
where $I_{\alpha}$ is defined as in (\ref{defI0}).
\end{thm}
\begin{proof}
We can repeat exactly what we did in the first two steps of the proof of Theorem \ref{scithm}. At this point, 
let us distinguish the cases $\alpha = 3$ and $\alpha >3$. 

\noindent
\textit{Case $\alpha = 3$.}

\noindent
Starting from (\ref{identig}), we can identify the tangential component of $g$. Indeed,
observe that, if we write
\begin{equation*}
\int_{D}\big(\partial_s Y^{(h)} - R^{(h)}\tau\big)\cdot \tau\,d\xi\,d\zeta = \,
\int_{D}\partial_s \big(Y^{(h)} - \Psi^{(h)}\big)\cdot\tau\,d\xi\,d\zeta -
\int_{D}\big(R^{(h)}\tau - \tau\big)\cdot\tau\,d\xi\,d\zeta,
\end{equation*}
 by the definition of $\big(u^{(h)}\big)$ we get
\begin{equation}\label{disting}
\frac{1}{h^{\alpha - 1}}\int_{D}\big(\partial_s Y^{(h)} - R^{(h)}\tau\big)\cdot \tau\,d\xi\,d\zeta = \, \big(u^{(h)}\big)' - \frac{1}{h^{\alpha - 1}}\int_{D}\big(R^{(h)}\tau - \,\tau\big)\cdot\tau\,d\xi\,d\zeta.
\end{equation}
If we let $h\rightarrow 0$ in (\ref{disting}) we obtain, from (\ref{identig}) and in virtue of 
property (f) in Theorem \ref{compactness},
\begin{equation}\label{use}
g\cdot \tau = u' - \frac{1}{2}\,\big(A^2\,\tau\big)\cdot \tau. 
\end{equation}
Notice that, using the explicit expression of $A$ given in (\ref{A}), we have
\begin{equation}\label{useful}
\frac{1}{2}\,\big(A^2\,\tau\big)\cdot \tau = - \,\frac{1}{2}\,\big((v'\cdot\nu_2)^2 + (v'\cdot\nu_3)^2\big).
\end{equation}
Now, by (\ref{vif}),(\ref{use}) and (\ref{useful}), we can write the expression 
of $\tilde{G}$ in this case, which turns to be
\begin{equation}\label{vif4}
\tilde{G} = R_0 \,\hbox{sym}\Bigg(\Big(B' + 2\,\hbox{skw}\big(R_0^T R'_0 B\big)\Big)\Bigg(\begin{array}{c}
0\\
\xi\\
\zeta
\end{array}\Bigg) + \bigg(u' + \frac{1}{2}\big((v'\cdot\nu_2)^2 + 
(v'\cdot\nu_3)^2\big) \bigg)\,e_1 \bigg|\,
\partial_{\xi} \varphi\,\bigg|\,
\partial_{\zeta} \varphi
 \Bigg)\,R_0^T . 
\end{equation}
Since $\varphi(s,\cdot)\in W^{1,2}(D;\mathbb{R}^{3})$ for a.e.
$s\in (0, L)$, using the definition of $Q^{0}$ the bound (\ref{sc})
becomes, as we claimed,
\begin{equation*}
\liminf_{h\rightarrow 0}\frac{1}{h^6}\int_{\widetilde{\Omega}_{h}}
W\Big(\big(\Psi^{(h)}\big)^{-1} (x),\nabla \tilde{y}^{(h)}(x)\Big)
dx \geq I_3(u,v,w),
\end{equation*}
with $I_3$ defined in (\ref{defI0}).

\textit{Case $\alpha > 3$.}

\noindent
If we let $h\rightarrow 0$ in (\ref{disting}) we obtain from (\ref{identig}) and in virtue of 
property (f) in Theorem \ref{compactness},
\begin{equation*}
g\cdot \tau = u'. 
\end{equation*} 
In fact, being $\alpha > 3$, it turns out that $\alpha - 1 < 2(\alpha - 2)$, so
\begin{equation*}
\textnormal{sym}\big(R^{(h)} - Id\big)/h^{\alpha - 1} \rightarrow 0 \quad \mbox{uniformly on } (0, L).
\end{equation*}
Now we can write down the expression of $\tilde{G}$ for $\alpha>3$, that is
\begin{equation}\label{sempl1i}
\tilde{G} = R_0 \,\hbox{sym}\Bigg(\Big(B' + 2\,\hbox{skw}\big(R_0^T R'_0 B\big)\Big)\Bigg(\begin{array}{c}
0\\
\xi\\
\zeta
\end{array}\Bigg) + u'\,e_1 \bigg|\,
\partial_{\xi} \varphi\,\bigg|\,
\partial_{\zeta} \varphi
 \Bigg)\,R_0^T . 
\end{equation}
Since $\varphi(s,\cdot)\in W^{1,2}(D;\mathbb{R}^{3})$ for a.e.
$s\in (0, L)$, using (\ref{sc}), (\ref{sempl1i}) and the definition of $Q^{0}$, we obtain exactly (\ref{fine2}), as we claimed.
\end{proof}


\section{Construction of the recovery sequences} 
In this section we show that the lower bounds obtained in Theorems \ref{scithm} and \ref{scithm2} are optimal. 
Also in this case the scalings $2< \alpha < 3$ and $\alpha\geq 3$ will be 
treated separately. However, we will first consider the higher scalings $h^{\alpha}$ with $\alpha \geq 3$, 
since as in  in \cite{FJM06}, the case $2< \alpha < 3$ turns out to be very delicate and requires a more 
detailed analysis.


\subsection{Higher scaling.} Let us consider the higher scalings of the energy, that is the case $\alpha \geq 3$.
\begin{thm}[Case $\alpha\geq 3$]\label{bfa}
For every $u, w \in W^{1,2}(0,L)$ and $v\in
W^{2,2}((0,L);\mathbb{R}^3)$ such that $v'\cdot \tau = 0$ there
exists a sequence $\big(\check{y}^{(h)}\big)\subset
W^{1,2}(\tilde{\Omega}_{h}; \mathbb{M}^{3\times 3})$  such that, setting 
$Y^{(h)}:= \check{y}^{(h)}\circ \Psi^{(h)}$, we have
\vspace{.15cm}
\begin{itemize}
\item[(i)] $\big(\nabla_{h}
Y^{(h)}\big(\nabla_{h}\Psi^{(h)}\big)^{-1}- Id\big)/h^{\alpha - 2}
\rightarrow A$ \, strongly in $L^{2}(\Omega;\mathbb{M}^{3\times
3})$;
\vspace{.2cm}
\item[(ii)] $v^{(h)}\rightarrow v$ \, strongly in $W^{1,2}((0, L);\mathbb{R}^{3})$;
\vspace{.2cm}
\item[(iii)] \,$w^{(h)}\rightharpoonup w$\, weakly in $W^{1,2}(0, L)$;
\vspace{.2cm}
\item[(iv)] $u^{(h)}\rightharpoonup u$\, weakly in $W^{1,2}(0, L)$,
\end{itemize}
where $A$, $v^{(h)}$, $u^{(h)}$, and $w^{(h)}$ are
defined as in (\ref{A}), (\ref{vh}), (\ref{uh}) and (\ref{wh}). Moreover,
\begin{equation}\label{limsup}
\limsup_{h\rightarrow 0}\frac{1}{h^{2\,\alpha}}\int_{\widetilde{\Omega}_{h}}
W\Big(\big(\Psi^{(h)}\big)^{-1} (x),\nabla
\check{y}^{(h)}(x)\Big) dx \leq I_{\alpha}(u, v, w),
\end{equation}
where $I_{\alpha}$ is defined in (\ref{defI0}).
\end{thm}
\begin{proof}
As first step we assume to deal with more regular functions; more precisely, we require 
that $u, w\in C^{1}[0,L]$ and $v\in C^{2}([0, L];\mathbb{R}^3)$. 

As in \cite{MGMM04}, let us define the functions  $\gamma_2,\gamma_3, \kappa^{(h)}: [0,
L]\rightarrow\mathbb{R}^3$ in the following way:
\begin{align}
\gamma_2(s):=&\,2\,w\,(v'\cdot\nu_3)\,e_1 + \big(w^2 + (v'\cdot\nu_2)^2\big)\,e_2 + (v'\cdot\nu_2)\,(v'\cdot\nu_3)\,e_3,\label{defg2}\\
\gamma_3(s):=&\,- 2\,w\,(v'\cdot\nu_2)\,e_1 + (v'\cdot\nu_2)\,(v'\cdot\nu_3)\,e_2  + \big(w^2 + (v'\cdot\nu_3)^2\big)\,e_3,\label{defg3}\\
\kappa^{(h)}(s,\xi,\zeta):=&\,(1 - h\,\xi\,k_2 - h\,\zeta\,k_3)\,\tau\label{defk},
\end{align}
where $k_2$ and $k_3$ are the scalar functions defined in (\ref{curv}). 
Finally, let $\varphi \in C^1(\bar{\Omega};\mathbb{R}^3)$ and let $\beta: \Omega \rightarrow \mathbb{R}^3$ be
\begin{equation}\label{chiari}
\beta(s,\xi,\zeta) := \left\{
\begin{array}{ll}
\vspace{.1cm}
\displaystyle R_0(s)\varphi(s,\xi,\zeta) -
\frac{1}{2}\,\xi\,R_{0}(s)\gamma_2(s) - \frac{1}{2}\,\zeta\,R_0(s)\gamma_3(s) & \mbox{if } \,\, \alpha = 3,\\
\displaystyle R_0(s)\varphi(s,\xi,\zeta)  & \mbox{if }\,\, \alpha > 3.
\end{array}
\right.
\end{equation}

For every $h > 0$ consider the function $Y^{(h)}:
\Omega\rightarrow \mathbb{R}^{3}$ defined as
\begin{equation}\label{defor}
Y^{(h)}=\, \Psi^{(h)} + h^{\alpha - 2}\,v + h^{\alpha - 1}\, u\,\kappa^{(h)} + h^{\alpha - 1} \xi\,A\,\nu_2  + \, h^{\alpha - 1} \zeta\,A\,\nu_3  
+ h^\alpha \beta, 
\end{equation}
where the matrix $A$ is defined as in (\ref{A}). 

Let us compute the scaled gradient of the deformation $Y^{(h)}$. First of all notice that 
$\nabla_{h}\kappa^{(h)} = \big(\tau'\,\big|\,- \big(\tau'\cdot \nu_2\big)\,\tau\,\big|\,- 
\big(\tau'\cdot \nu_3\big)\,\tau\big) + O(h)$, and that 
\begin{equation}\label{antisim}
\Big(\tau'\,\big|\,- \big(\tau'\cdot \nu_2\big)\,\tau\,\big|\,- \big(\tau'\cdot \nu_3\big)\,\tau\Big)
= \big(\tau'\otimes\tau - \tau\otimes\tau'\big)\,R_0.
\end{equation}

Hence, the scaled gradient turns to be 
\begin{align}\label{gradYh}
\nabla_{h}Y^{(h)} &=\, \nabla_{h}\Psi^{(h)} + h^{\alpha-2}\,A\,R_0 + h^{\alpha-1}\,\Bigg(\big(A\,R_0\big)'\Bigg(\begin{array}{c}
0\\
\xi\\
\zeta
\end{array}\Bigg) + \,u'\,\tau \,\bigg|\,
\partial_{\xi} \beta\,\bigg|\,
\partial_{\zeta} \beta
 \Bigg) \,+ \nonumber\\  
& + h^{\alpha-1}\,u\,\big(\tau'\otimes\tau - \tau\otimes\tau'\big)\,R_0 +  O(h^\alpha). 
\end{align}

So we have that, by (\ref{invA}),
\begin{equation*}
\frac{1}{h^{\alpha - 2}}\Big(\nabla_{h} Y^{(h)}\big(\nabla_{h}\Psi^{(h)}\big)^{-1} - Id\Big) = A + O(h)
\end{equation*}
and this proves (i). Now remark that, if we
define $v^{(h)}$ as in (\ref{vh}), we have, using (\ref{dom2}),
\begin{equation*} 
v^{(h)} = 
v + h \,u\,\tau + h^2 \int_{D}\beta\,d\xi\,d\zeta,
\end{equation*}
so also (ii) follows. For the sequence $w^{(h)}$ defined as in (\ref{wh}) we get, by (\ref{dom1}) and (\ref{dom2}),
\begin{align*}
w^{(h)}  
&=\frac{1}{\mu(D)}\int_{D}(\xi\,A\,\nu_2 + \zeta\,A\,\nu_3 + h\,\beta) 
\cdot (\xi\,\nu_{3} - \zeta\,\nu_{2})\,d\xi\,d\zeta\\
&= (A\,\nu_2)\cdot \nu_3 + O(h),
\end{align*}
which is exactly $w$, up to a perturbation of order $h$. This proves (iii).

Moreover, if we define $u^{(h)}$ as in (\ref{uh}) we have
\begin{align}\label{compu}
\big(u^{(h)}\big)' = \frac{1}{h}\int_D\big(v'\cdot\tau + h\,u' + h^2 \partial_s\beta \cdot \tau\big)\,d\xi\,d\zeta 
=\,u' + h\int_D \partial_s\beta \cdot \tau\,d\xi\,d\zeta
\end{align}
hence the convergence property in (iv) is also proved.

Once all these properties are satisfied, we can show (\ref{limsup}). Using (\ref{invA}) and (\ref{gradYh})
we have
\begin{align}\label{YPsi}
Z^{(h)}:=\,\nabla_{h} Y^{(h)}\big(\nabla_{h}\Psi^{(h)}\big)^{-1} &=\, Id + h^{\alpha-2}\Bigg(A + h\,u'\,\tau\otimes\tau +
h\,\Bigg(A'\,R_0\Bigg(\begin{array}{c}
0\\
\xi\\
\zeta
\end{array}\Bigg) \bigg|\,
\partial_{\xi} \beta\,\bigg|\,
\partial_{\zeta} \beta
 \Bigg)\,R_0^T\Bigg) \nonumber\\
&+\, h^{\alpha-1}\,u\,\big(\tau'\otimes\tau - \tau\otimes\tau'\big) + O(h^\alpha).
\end{align}
Using the identity $(Id + B^T)(Id + B) = Id + 2\,\mbox{sym}\,B +
B^T B$, we obtain for the nonlinear
strain
\begin{align}\label{fundamental}
\big(Z^{(h)}\big)^T Z^{(h)} =& \, Id + \,2\,h^{\alpha-1}\,\mbox{sym}\, \Bigg( \Bigg(A'\,R_0\Bigg(\begin{array}{c}
0\\
\xi\\
\zeta
\end{array}\Bigg) \bigg|\,
\partial_{\xi} \beta\,\bigg|\,
\partial_{\zeta} \beta
 \Bigg)\,R_0^T \Bigg)
+ \, 2\,h^{\alpha-1}\,u'\,\tau\otimes\tau \nonumber\\
+& \,h^{2(\alpha-2)}A^T A \,+\, o\big(h^{2(\alpha-2)}\big),
\end{align}
where $o(h^{\gamma})/h^\gamma \rightarrow 0$ uniformly as $h\rightarrow 0$.

Now, let us distinguish the cases $\alpha = 3$ and $\alpha > 3$.

\noindent
\textit{Case $\alpha = 3$.}

Notice that if we specify $\alpha = 3$ in (\ref{fundamental}), all the terms
are of the same order with respect to $h$, that is of order $2$.
Taking the square root we have that
\begin{equation}\label{root}
\left[\big(Z^{(h)}\big)^T Z^{(h)}\right]^{1/2} = Id + h^2
\tilde{G} + O(h^3),
\end{equation}
where
\begin{equation*}
\tilde{G} := u'\tau\otimes\tau + \textnormal{sym}\,\Bigg(
\Bigg(A'\,R_0\Bigg(\begin{array}{c}
0\\
\xi\\
\zeta
\end{array}\Bigg) \bigg|\,
\partial_{\xi} \beta\,\bigg|\,
\partial_{\zeta} \beta
 \Bigg)\,R_0^T \Bigg) -\frac{A^2}{2}.
\end{equation*}
In order to write $\tilde{G}$ in a more useful way, notice that, by (\ref{A}),
\begin{align}\label{sempl4}
\tilde{G} =&\, R_0\Bigg[ \textnormal{sym}\,\Bigg( R_0^T
A'\,R_0\Bigg(\begin{array}{c}
0\\
\xi\\
\zeta
\end{array}\Bigg) + \bigg(u' + \frac{1}{2}\big((v'\cdot\nu_2)^2 + (v'\cdot\nu_3)^2\big) \bigg)\,e_1 \bigg|\,
\partial_{\xi} (R_0^T\beta)\,\bigg|\,
\partial_{\zeta} (R_0^T\beta)
\Bigg)\Bigg]\,R_{0}^T  \nonumber\\
&+\,\frac{1}{2}\,R_0 \left(\begin{array}{c} \vspace{.15cm}
0\\
\vspace{.15cm}
w\,(v'\cdot\nu_3)\\
-w\,(v'\cdot\nu_2)
\end{array}
\begin{array}{c}
\vspace{.15cm}
w\,(v'\cdot\nu_3)\\
\vspace{.15cm}
w^2 + (v'\cdot\nu_2)^2\\
(v'\cdot\nu_2)\,(v'\cdot\nu_3)
\end{array}
\begin{array}{c}
\vspace{.15cm}
-w\,(v'\cdot\nu_2)\\
\vspace{.15cm}
(v'\cdot\nu_2)\,(v'\cdot\nu_3)\\
w^2 + (v'\cdot\nu_3)^2
\end{array} \right)\,R_0^T.
\end{align}
We can rewrite (\ref{sempl4}) in terms of $\varphi$ 
and $B$, using (\ref{vifor}) and (\ref{chiari}), as
\begin{equation*}
\tilde{G} = \textnormal{sym}\Bigg[R_0\,\Bigg(\Big(B' + 2\,\hbox{skw}\big(R_0^T R'_0 B\big)\Big)
\Bigg(\begin{array}{c}
0\\
\xi\\
\zeta
\end{array}\Bigg) + \bigg(u' + \frac{1}{2}\big((v'\cdot\nu_2)^2 + (v'\cdot\nu_3)^2\big)\bigg)\,e_1 \,\bigg|
\,\partial_{\xi} \varphi\,\bigg| \,\partial_{\zeta}\varphi
\Bigg)\,R_{0}^T\Bigg].
\end{equation*}
From the frame-indifference of the energy density $W$, since
$\mbox{det}\,\big(\nabla_{h}
Y^{(h)}\big)\big(\nabla_{h}\Psi^{(h)}\big)^{-1} > 0$ for
sufficiently small $h$, we have
\begin{equation*}
W\big(s,\xi,\zeta,Z^{(h)}\big) =
W\Big(s,\xi,\zeta,\big[\big(Z^{(h)}\big)^T Z^{(h)}\big]^{1/2}\Big).
\end{equation*}
Thus, by (\ref{root}) and Taylor expansion, we obtain
\begin{equation*}
\frac{1}{h^4}\,W\big(s,\xi,\zeta,Z^{(h)}\big) \rightarrow
\frac{1}{2} \,Q_{3}(s,\xi,\zeta,\tilde{G}) \,\,
\mbox{a.e.,}
\end{equation*}
and
\begin{equation*}
\frac{1}{h^4}\,W\big(s,\xi,\zeta,Z^{(h)}\big) \leq
\frac{1}{2}\,\gamma\,|\,\tilde{G}\,|^2 + C\,h \leq
C\,\big(|\,B\,|^4 + |\,B'|^2 + |\,\partial_\xi\varphi\,|^2 +
|\,\partial_\zeta\varphi\,|^2  + |\,u'|^2 + 1\big)\in L^1(\Omega).
\end{equation*}
Set $\check{y}^{(h)}:= Y^{(h)}\circ \big(\Psi^{(h)}\big)^{-1}$; by the dominated convergence theorem we get the following
equality:
\begin{equation}\label{quasifin}
\limsup_{h\rightarrow 0}\frac{1}{h^6}\int_{\widetilde{\Omega}_{h}}
W\Big(\big(\Psi^{(h)}\big)^{-1} (x),\nabla \check{y}^{(h)}(x)\Big)
dx =
\frac{1}{2}\int_{\Omega}Q_3(s,\xi,\zeta,\tilde{G})\,ds\,d\xi\,d\zeta.
\end{equation}

Consider the general case. Let $u, w\in W^{1,2}(0,L)$ and $v\in W^{2,2}((0,L);\mathbb{R}^3)$.
Let $\varphi(s,\cdot)\in \mathcal{V}$ be the solution of the minimum problem (\ref{defQ0}) defining $Q^0$, with 
$t:= u' + \frac{1}{2}\,\big((v'\cdot \nu_2)^2 + (v'\cdot \nu_3)^2\big)$ and $F:= B' + 2\,\hbox{skw}\big(R_0^T R'_0 B\big)$, where $B$ is introduced in (\ref{An}). As we have already noticed in Remark \ref{rk}, $\varphi$ and its derivatives with respect to $\xi$ and $\zeta$ belong to $L^2(\Omega;\mathbb{R}^3)$. 

Now, we can smoothly approximate $u, w$ in the strong topology of $W^{1,2}$, $v$ in the strong topology of 
$W^{2,2}$, and $\varphi$, $\partial_{\xi}\varphi$ and $\partial_{\zeta}\varphi$ in the strong topology 
of $L^2$. Since the approximating sequences satisfy (\ref{quasifin}), and the right-hand side of (\ref{quasifin}) is continuous with respect to the mentioned topologies, we conclude that (\ref{quasifin}) holds also in the general case.
Hence, using the minimality of $\varphi$, we obtain (\ref{limsup}).

\textit{Case $\alpha > 3$.}

\noindent
In this case, in the expression (\ref{fundamental}), the term of order $2(\alpha -2)$ in $h$ 
can be neglected, since $2(\alpha -2) > \alpha -1$ when $\alpha>3$. Hence we can write
\begin{align*}
\big(Z^{(h)}\big)^T Z^{(h)} = Id + \,2\,h^{\alpha-1}\,\mbox{sym}\, \Bigg( \Bigg(A'\,R_0\Bigg(\begin{array}{c}
0\\
\xi\\
\zeta
\end{array}\Bigg) \bigg|\,
\partial_{\xi} \beta\,\bigg|\,
\partial_{\zeta} \beta
 \Bigg)\,R_0^T \Bigg)
+ \, 2\,h^{\alpha-1}\,u'\,\tau\otimes\tau + o(h^{\alpha-1}),
\end{align*}
where $o(h^{\gamma})/h^\gamma \rightarrow 0$ uniformly as $h\rightarrow 0$.
Taking the square root we have that
\begin{equation}\label{rooti}
\left[\big(Z^{(h)}\big)^T Z^{(h)}\right]^{1/2} = Id + h^{\alpha-1}
\tilde{G} + o\big(h^{\alpha-1}\big),
\end{equation}
where
\begin{equation*}
\tilde{G} := u'\tau\otimes\tau + \textnormal{sym}\,\Bigg(
\Bigg(A'\,R_0\Bigg(\begin{array}{c}
0\\
\xi\\
\zeta
\end{array}\Bigg) \bigg|\,
\partial_{\xi} \beta\,\bigg|\,
\partial_{\zeta} \beta
 \Bigg)\,R_0^T \Bigg).
\end{equation*}
We can rewrite $\tilde{G}$ in terms of $\varphi$ and $B$ as
\begin{equation*}
\tilde{G} = \textnormal{sym}\Bigg[R_0\,\Bigg(\Big(B' + 2\,\hbox{skw}\big(R_0^T R'_0 B\big)\Big)
\Bigg(\begin{array}{c}
0\\
\xi\\
\zeta
\end{array}\Bigg) + u'\,e_1 \,\bigg|
\,\partial_{\xi} \varphi\,\bigg| \,\partial_{\zeta}\,\varphi
\Bigg)\,R_{0}^T\Bigg].
\end{equation*}
From the frame-indifference of the energy density $W$, since
$\mbox{det}\,\big(\nabla_{h}
Y^{(h)}\big)\big(\nabla_{h}\Psi^{(h)}\big)^{-1} > 0$ for
sufficiently small $h$, we have
\begin{equation*}
W\big(s,\xi,\zeta,Z^{(h)}\big) =
W\Big(s,\xi,\zeta,\big[\big(Z^{(h)}\big)^T Z^{(h)}\big]^{1/2}\Big);
\end{equation*}
thus, by (\ref{rooti}) and Taylor expansion, we obtain
\begin{equation*}
\frac{1}{h^{2\,\alpha -2}}\,W\big(s,\xi,\zeta,Z^{(h)}\big) \rightarrow
\frac{1}{2} \,Q_{3}(s,\xi,\zeta,\tilde{G}) \,\,
\mbox{a.e.,}
\end{equation*}
and
\begin{equation*}
\frac{1}{h^{2\,\alpha -2}}\,W\big(s,\xi,\zeta,Z^{(h)}\big) \leq
\frac{1}{2}\,\gamma\,|\,\tilde{G}\,|^2 + C\,h \leq
C\,\big(|\,B'|^2 + |\,B\,|^2 + |\,\partial_\xi\varphi\,|^2 +
|\,\partial_\zeta \varphi\,|^2  + |\,u'|^2 + 1\big)\in L^1(\Omega).
\end{equation*}
Set $\check{y}^{(h)}:= Y^{(h)}\circ \big(\Psi^{(h)}\big)^{-1}$; by the dominated convergence theorem we get the following
equality:
\begin{equation}\label{fin}
\limsup_{h\rightarrow 0}\frac{1}{h^{2\,\alpha}}\int_{\widetilde{\Omega}_{h}}
W\Big(\big(\Psi^{(h)}\big)^{-1} (x),\nabla \check{y}^{(h)}(x)\Big)
dx =
\frac{1}{2}\int_{\Omega}Q_3(s,\xi,\zeta,\tilde{G})\,ds\,d\xi\,d\zeta.
\end{equation}

Consider the general case. Let $u, w\in W^{1,2}(0,L)$ and $v\in W^{2,2}((0,L);\mathbb{R}^3)$.
Let $\varphi(s,\cdot)\in \mathcal{V}$ be the solution of the minimum problem (\ref{defQ0}) defining $Q^0$, with 
$t:= u'$ and $F:= B' + 2\,\hbox{skw}\big(R_0^T R'_0 B\big)$, where $B$ is defined as in (\ref{An}). 
It is easy to show that (\ref{fin}) remains true, following the same approximation arguments used in the previous step. Hence, using the minimality of $\varphi$, we obtain (\ref{limsup}).
\end{proof}


\subsection{Intermediate scaling}
We now consider the scalings $h^{\alpha}$ with $2< \alpha < 3$. As in \cite{FJM06}, 
this case turns out to be very delicate and requires a detailed analysis.

\begin{thm}[Case $2<\alpha<3$]\label{Gamcon3}
For every $w \in W^{1,2}(0,L)$ and $v\in W^{2,2}((0,L);\mathbb{R}^3)$ such that $v'\cdot \tau = 0$ 
there exists a sequence $\big(\check{y}^{(h)}\big)\subset W^{1,2}(\tilde{\Omega}_{h}; \mathbb{M}^{3\times 3})$ 
such that, setting $Y^{(h)}:= \check{y}^{(h)}\circ \Psi^{(h)}$, we have  
\begin{itemize}
\item[(i)] $\Big(\big(\nabla_{h}
Y^{(h)}\big)\big(\nabla_{h}\Psi^{(h)}\big)^{-1}- Id\Big)/h^{\alpha - 2}
\rightarrow A$ \, strongly in $L^{2}(\Omega;\mathbb{M}^{3\times
3})$;
\item[(ii)] $v^{(h)}\rightarrow v$ \, strongly in $W^{1,2}((0, L);\mathbb{R}^{3})$;
\vspace{.15cm}
\item[(iii)] $w^{(h)}\rightharpoonup w$\, weakly in $W^{1,2}(0, L)$,
\end{itemize}
with $A, v^{(h)}$ and $w^{(h)}$ defined as in (\ref{A}), (\ref{vh}) and (\ref{wh}). Moreover
\begin{equation}\label{limsupi3}
\limsup_{h\rightarrow 0}\frac{1}{h^{2\,\alpha}}\int_{\widetilde{\Omega}_{h}}
W\Big(\big(\Psi^{(h)}\big)^{-1} (x),\nabla
\check{y}^{(h)}(x)\Big) dx \leq I_\alpha(v, w),
\end{equation}
where $I_\alpha$ is introduced in (\ref{defI}).
\end{thm}


\begin{proof}
As in Theorem \ref{bfa}, we preliminarly assume that $w\in C^1[0,L]$ and 
$v\in C^2([0,L];\mathbb{R}^3)$. Let $g \in C^0[0,L]$ and $\varphi \in C^1(\bar{\Omega};\mathbb{R}^3)$. 
Denote by $\beta$ the function $\beta(s,\xi,\zeta) := R_0(s)\varphi(s,\xi,\zeta)$ and by $\tilde{g}$ 
a primitive of the function $g$.

Define the functions  $\gamma_2,\gamma_3, \kappa^{(h)}$  as in the proof of Theorem \ref{bfa}. 
Finally define the function
$u\in C^{1}[0,L]$ as a primitive of
\begin{equation*}
 - \,\frac{1}{2}\Big((v'\cdot \nu_2)^2 + (v'\cdot \nu_3)^2\Big).
\end{equation*}

\noindent
In analogy with the cases $\alpha \geq 3$, one could make the ansatz
\begin{align}\label{ansatz}
Y^{(h)}&=\,\Psi^{(h)} + \,h^{\alpha - 2}v +\, h^{\alpha - 1}\xi\,A\,\nu_2
+\, h^{\alpha - 1}\zeta\,A\,\nu_3 + \big( h^{2(\alpha - 2)}u + h^{\alpha - 1}\tilde{g}\big)\,
\kappa^{(h)}\,  +\nonumber\\
&\,- \frac{1}{2}\,h^{(2\,\alpha - 3)}R_0 \big(\xi\,\gamma_2 + \zeta\,\gamma_3\big) + h^\alpha \beta.
\end{align}
 Hence, by (\ref{antisim}) the scaled gradient of the deformation $Y^{(h)}$ is 
\begin{align}\label{gradYhi}
&\nabla_{h}Y^{(h)} =\, \nabla_{h}\Psi^{(h)} + h^{\alpha-2}\,A\,R_0 + h^{\alpha-1}\,\Bigg(\big(A\,R_0\big)'\Bigg(\begin{array}{c}
0\\
\xi\\
\zeta
\end{array}\Bigg) \,\bigg|\,
\partial_{\xi} \beta\,\bigg|\,
\partial_{\zeta} \beta
 \Bigg)\,+ \, h^{\alpha - 1}g\,\tau\otimes e_1 \,+ \nonumber\\  
&+\,\big(h^{2(\alpha - 2)}u + h^{\alpha - 1}\tilde{g}\big)
\,\big(\tau'\otimes\tau - \tau\otimes\tau'\big)\,R_0\,+\,\frac{1}{2}\,h^{2(\alpha - 2)}\,R_0 \,\big(2\,u'\,e_1\,|\,- \gamma_2\,|\,- \gamma_3\big) + o(h^{\alpha-1}).
\end{align}
Now, using (\ref{invA}) and (\ref{gradYhi}) we have
\begin{align}\label{YPsii3}
Z^{(h)}&:=\, \nabla_{h} Y^{(h)}\big(\nabla_{h}\Psi^{(h)}\big)^{-1} =\, Id + h^{\alpha-2}A +
h^{\alpha-1}\,\Bigg(A'\,R_0\Bigg(\begin{array}{c}
0\\
\xi\\
\zeta
\end{array}\Bigg)\,+ g\,\tau \bigg|\,
\partial_{\xi} \beta\,\bigg|\,
\partial_{\zeta} \beta  \Bigg)\,R_0^T +\nonumber\\
&+\,\big(h^{2(\alpha - 2)}u + h^{\alpha - 1}\tilde{g}\big)\,\big(\tau'\otimes \tau - \tau \otimes \tau'\big)  + \, \frac{1}{2}\,h^{2(\alpha - 2)}\,R_0 \,\big(2\,u'\,e_1\,|\,- \gamma_2\,|\,- \gamma_3\big)\,R_0^T + o(h^{\alpha - 1}).
\end{align}
This procedure leads to the desired conclusion for $\alpha > 5/2$, but our ansatz cannot work for 
$\alpha$ close to $2$. Indeed, for $\alpha > 5/2$, using the identity $(Id + P^T)(Id + P) = Id + 2\,\mbox{sym}\,P +
P^T P$, and noticing that some of the matrices on the right-hand side
of (\ref{YPsii3}) are skew-symmetric, we obtain for the nonlinear
strain
\begin{align}\label{strain}
\big(Z^{(h)}\big)^T Z^{(h)} =&\, Id + \,2\,h^{\alpha-1}\,\mbox{sym}\, \Bigg( \Bigg(A'\,R_0\Bigg(\begin{array}{c}
0\\
\xi\\
\zeta
\end{array}\Bigg) \bigg|\,
\partial_{\xi} \beta\,\bigg|\,
\partial_{\zeta} \beta
 \Bigg)\,R_0^T \Bigg)
+ \, 2\,h^{\alpha-1}\,g\,\tau\otimes\tau \nonumber\\
+&\,h^{2(\alpha - 2)}R_0 \,\Big(\mbox{sym}\,(2\,u'\,e_1\,|\,- \gamma_2\,|\,- \gamma_3)\Big)\,R_0^T + 
h^{2(\alpha - 2)}A^T A + o(h^{\alpha-1}).
\end{align}
Moreover, using (\ref{A}) and our definition of $u$, $\gamma_2$ and $\gamma_3$, we have that
\begin{equation}\label{rooti3}
\left[\big(Z^{(h)}\big)^T Z^{(h)}\right]^{1/2} = Id + h^{\alpha-1}
\tilde{G} + o\big(h^{\alpha-1}\big),
\end{equation}
where
\begin{equation*}
\tilde{G} := g\,\tau\otimes\tau + \textnormal{sym}\,\Bigg(
\Bigg(A'\,R_0\Bigg(\begin{array}{c}
0\\
\xi\\
\zeta
\end{array}\Bigg) \bigg|\,
\partial_{\xi} \beta\,\bigg|\,
\partial_{\zeta} \beta
 \Bigg)\,R_0^T \Bigg).
\end{equation*}
As in Theorem \ref{bfa}, the frame-indifference of the energy density $W$ and 
the dominated convergence theorem give the following
equality:
\begin{equation}\label{quasi}
\limsup_{h\rightarrow 0}\frac{1}{h^{2\,\alpha}}\int_{\widetilde{\Omega}_{h}}
W\Big(\big(\Psi^{(h)}\big)^{-1} (x),\nabla \check{y}^{(h)}(x)\Big)
dx =
\frac{1}{2}\int_{\Omega}Q_3(s,\xi,\zeta,\tilde{G})\,ds\,d\xi\,d\zeta,
\end{equation}
and the general case can be proved by approximation. 
Then, using the minimality assumptions on $g$ and $\varphi$, we obtain 
(\ref{limsupi3}) and so the claim.

Unfortunately, this procedure fails for $\alpha$ close to $2$, since in that case 
terms of order $h^{4(\alpha - 2)}$ appear in the expression of the nonlinear strain 
$\big(Z^{(h)}\big)^T Z^{(h)}$, and they cannot be absorbed in $o(h^{\alpha - 1})$. 

Therefore, in the spirit of the proof of \cite[Theorem 6.2]{FJM06}, we modify the ansatz 
(\ref{ansatz}) in order to get an exact isometry.
Let us define for every $h > 0$, the sequence 
\begin{equation}\label{ansatz2}
Y^{(h)}:= \int_{0}^{s} (R_\varepsilon \tau)\,d\sigma + h\xi R_\varepsilon \nu_2 + 
h\zeta R_\varepsilon \nu_3 + h^\alpha \beta,
\end{equation}
where $R_{\varepsilon}:= e^{\varepsilon A}$, with $A$ defined as in (\ref{A}), 
and $\varepsilon:= h^{\alpha - 2}$.
Notice that, due to the fact that $A$ is skew-symmetric, the matrix $R_\varepsilon$ 
turns out to be a rotation.

The scaled gradient of the deformation $Y^{(h)}$ is given by 
\begin{equation}\label{una}
\nabla_h Y^{(h)} = R_\varepsilon R_0 + h\,R_\varepsilon(\xi \nu_2' + \zeta \nu_3')\otimes e_1 + 
h\,\Bigg(R_{\varepsilon}'\,R_0\Bigg(\begin{array}{c}
0\\
\xi\\
\zeta
\end{array}\Bigg) \bigg|\,
h^{\alpha - 2} \partial_{\xi} \beta\,\bigg|\,
h^{\alpha - 2} \partial_{\zeta} \beta  \Bigg) + O(h^{\alpha}).
\end{equation}
Now, using (\ref{psi}), the expression (\ref{una}) becomes 
\begin{equation*}
\nabla_h Y^{(h)} = R_\varepsilon \nabla_h \Psi^{(h)} + 
h\,\Bigg(R_{\varepsilon}'\,R_0\Bigg(\begin{array}{c}
0\\
\xi\\
\zeta
\end{array}\Bigg) \bigg|\,
h^{\alpha - 2} \partial_{\xi} \beta\,\bigg|\,
h^{\alpha - 2} \partial_{\zeta} \beta  \Bigg) + O(h^{\alpha}),
\end{equation*}
and hence, by (\ref{invA}) we have 

\begin{align}\label{zeta}
Z^{(h)}&:=\, \nabla_h Y^{(h)}(\nabla_h \Psi^{(h)})^{-1} = R_\varepsilon + 
h\,\Bigg(R_{\varepsilon}'\,R_0\Bigg(\begin{array}{c}
0\\
\xi\\
\zeta
\end{array}\Bigg) \bigg|\,
h^{\alpha - 2} \partial_{\xi} \beta\,\bigg|\,
h^{\alpha - 2} \partial_{\zeta} \beta  \Bigg)R_0^T + o(h^{\alpha -1}),\nonumber \\
&= \,R_\varepsilon\,\left(Id + 
h\,\Bigg(R_{\varepsilon}^T R_{\varepsilon}'\,R_0\Bigg(\begin{array}{c}
0\\
\xi\\
\zeta
\end{array}\Bigg) \bigg|\,
h^{\alpha - 2} \partial_{\xi} (R_{\varepsilon}^T\beta)\,\bigg|\,
h^{\alpha - 2} \partial_{\zeta} (R_{\varepsilon}^T\beta)  \Bigg)R_0^T \right)+ o(h^{\alpha -1}).
\end{align}
Now notice that, by definition, the rotation $R_{\varepsilon}$ verifies the identities:
\begin{equation*}
R_{\varepsilon}(s) = Id + \varepsilon\, A(s) + o (\varepsilon), \quad 
R'_{\varepsilon}(s)  = \varepsilon \int_0^1 e^{(1 - \sigma)\varepsilon\,A(s)} A'(s)\,e^{\sigma\,\varepsilon\,A(s)}d\sigma.
\end{equation*}
Therefore we have in particular 
\begin{equation}\label{prop}
R_{\varepsilon}^T R_{\varepsilon}' = \varepsilon\,A' + o(\varepsilon).
\end{equation}
Hence, using (\ref{prop}) and the fact that $\varepsilon = h^{\alpha - 2}$, (\ref{zeta}) simplifies as follows
\begin{equation*}
Z^{(h)} = \,R_\varepsilon\,\left(Id + 
h^{\alpha - 1}\Bigg(A'\,R_0\Bigg(\begin{array}{c}
0\\
\xi\\
\zeta
\end{array}\Bigg) \bigg|\, \partial_{\xi} \beta\,\bigg|\,
\partial_{\zeta} \beta  \Bigg)R_0^T \right)+ o(h^{\alpha -1}).
\end{equation*}
Thus, using frame indifference, we obtain 
\begin{align*}
\frac{1}{h^{2\,\alpha -2}}\,W\big(s,\xi,\zeta,Z^{(h)}\big) &= 
\frac{1}{h^{2\,\alpha -2}}\,W\Big(s,\xi,\zeta,\big(R_{\varepsilon}\big)^T Z^{(h)}\Big)\\
&\rightarrow \frac{1}{2} \,Q_{3}(s,\xi,\zeta,\tilde{G}) \,\,
\mbox{a.e.,}
\end{align*} 
and proceeding as before we get the desired claim.
\end{proof}


\section{The case of a closed thin beam}

It appears natural to ask whether the same analysis that we have developed so far can be extended 
to the case of a thin rod whose mid-fiber is a closed curve. 
In this section we will show that this additional requirement imposes a restriction on the 
class of admissible limit deformations, while the expression of the limiting functional is not 
affected by this constraint.

Throughout this section we will assume $\alpha = 3$ for simplicity, but the results can be easily extended to
the other cases. 

The setting of the problem is exactly the same as before. The additional assumptions are
\begin{equation}\label{bc} 
\gamma(0) = \gamma(L), \gamma'(0) = \gamma'(L) \quad \hbox{ and } \nu_k(0) = \nu_k(L), \hbox{for } k = 2,3.
\end{equation}
Notice that, from (\ref{bc}) it easily follows that $\Psi^{(h)}(0,\xi,\zeta) = \Psi^{(h)}(L,\xi,\zeta)$ for every $(\xi,\zeta)\in D$.

Now we will state and prove a compactness result which allows to identify the domain of the $\Gamma$-limit.

\begin{thm}\label{compbc}
Let $\big(\tilde{y}^{(h)}\big)\subset W^{1,2}\big(\widetilde{\Omega}_{h};\mathbb{R}^{3}\big)$  be a sequence verifying
\begin{equation}
\frac{1}{h^4}\,\tilde{I}^{(h)}(\tilde{y}^{(h)})
\leq c < +\infty
\end{equation}
for every $h > 0$. Then there exist an associated sequence $R^{(h)}\subset C^{\infty}((0, L);\mathbb{M}^{3\times 3})$
and constants $\bar{R}^{(h)} \in SO(3)$, $c^{(h)} \in \mathbb{R}^3$ such that, if we define $Y^{(h)}:= \big(\bar{R}^{(h)}\big)^{T}\,\tilde{y}^{(h)}\circ \Psi^{(h)} - c^{(h)}$, we have
\begin{align}
& R^{(h)}(s) \in SO(3) \quad \hbox{for every } s\in (0,L), \label{rota}\\
\vspace{.18cm}
&\big|\big|R^{(h)} - Id\big|\big|_{L^{\infty}(0, L)} \leq C\,h,\quad 
\big|\big|\big(R^{(h)}\big)'\big|\big|_{L^{2}(0, L)} < C\,h,\label{rota1}\\
\vspace{.18cm}
&\big|\big|\nabla_{h} Y^{(h)}\big(\nabla_{h}\Psi^{(h)}\big)^{-1} - R^{(h)}\big|\big|_{L^{2}(\Omega)} \leq C\,h^2,\label{rota2}\\
\vspace{.18cm}
& \, \big| R^{(h)}(0) - R^{(h)}(L)\big| \leq c\,h^{3/2} \label{rota3}.
\end{align}
Moreover, defining $v^{(h)}$, $w^{(h)}$ and $u^{(h)}$ as in  (\ref{vh}), (\ref{wh}) and (\ref{uh}), we
have that, up to subsequences, the following properties are satisfied:
\vspace{.2cm}
\begin{itemize}
\item[(a)] \,$v^{(h)} \rightarrow v$ \, strongly in $W^{1,2}((0, L);\mathbb{R}^{3})$; moreover, $v\in W^{2,2}((0, L);\mathbb{R}^{3})$, $v'\cdot \tau = 0$, $v(0) = v(L)$, and $v'(0) = v'(L)$;
\vspace{.15cm}
\item[(b)]\,$w^{(h)} \rightharpoonup w$ \, weakly in $W^{1,2}(0, L)$, with $w(0) = w(L)$;
\vspace{.15cm}
\item[(c)]\,$u^{(h)}\rightharpoonup u$ \, weakly in $W^{1,2}(0, L)$; 
\vspace{.15cm}
\item[(d)] $\big(\nabla_{h} Y^{(h)}\,\big(\nabla_{h}\Psi^{(h)}\big)^{-1} - Id\,\big)/h \rightarrow A$
strongly in $L^2(\Omega;\mathbb{M}^{3\times 3})$, where the matrix $A\in W^{1,2}((0, L);\mathbb{M}^{3\times 3})$ is defined in (\ref{A});
\vspace{.15cm}
\item[(e)] $\big(R^{(h)} - Id\big)/h \rightharpoonup A$ weakly in $W^{1,2}((0, L);\mathbb{M}^{3\times 3})$;
\vspace{.15cm}
\item[(f)] $sym\big(R^{(h)} - Id\big)/h^2 \rightarrow A^2 /2$ \, uniformly on $(0, L)$.
\end{itemize} 
\end{thm}

\begin{proof}
The argument follows the proof of Proposition 4.1 in \cite{MGMS06}, but we will include the 
details for the convenience of the reader.
As in the proof of Theorem \ref{compactness}, the rigidity theorem provides the existence of a sequence of 
piecewise constant rotations $Q^{(h)}: (0,L) \rightarrow SO(3)$ such that, for every small cylinder $\widetilde{C}_{a,h}$ 
we have 
\begin{equation*}
\int_{\widetilde{C}_{a,h}}\big|\,\nabla\tilde{y}^{(h)}
- Q^{(h)}\big|^{2} dx \leq c
\int_{\widetilde{C}_{a,h}}\mbox{dist}^{2}(\nabla\tilde{y}^{(h)},SO(3))
dx.
\end{equation*}
Changing variables, the previous inequality becomes
\begin{equation}\label{rigid3}
\int_{\big(a, a + \frac{L}{K_{h}}\big)\times
D}\big|\,\nabla\tilde{y}^{(h)}\circ\Psi^{(h)} -Q^{(h)}\big|^{2} ds\,d\xi\,d\zeta 
\leq
c\,\int_{\big(a, a + \frac{L}{K_{h}}\big)\times
D}\mbox{dist}^{2}\big(\nabla\tilde{y}^{(h)}\circ\Psi^{(h)},SO(3)\big)ds\,d\xi\,d\zeta.
\end{equation}
Let us define $\bar{Q}:= Q^{(h)}(0)$. If we specify the relation for $a = 0$ we have
\begin{align}
\int_{\big(0, \frac{L}{K_{h}}\big)\times
D}\big|\,\nabla\tilde{y}^{(h)}\circ\Psi^{(h)} -
\bar{Q}\big|^{2} ds\,d\xi\,d\zeta &\leq c
\int_{\big(0, \frac{L}{K_{h}}\big)\times
D}\mbox{dist}^{2}(\nabla\tilde{y}^{(h)}\circ\Psi^{(h)},SO(3))ds\,d\xi\,d\zeta. \label{cil1}\\
\end{align}
In order to establish (\ref{rota3}), we start from the trace inequality
\begin{equation*}
\int_{D}|\,v(0,\xi,\zeta) - \bar{v}\,|^2\,d\xi\,d\zeta \leq c \int_{(0,l)\times D}|\,\nabla v\,|^2\,ds\,d\xi\,d\zeta,
\end{equation*}
which holds uniformly for $1 \leq l \leq 2$, with $\bar{v} = \int_{D}v(0,\xi,\zeta)\,d\xi\,d\zeta$. 
If we write this estimate for
\begin{equation*}
 v(s,\xi,\zeta):= \frac{1}{h}\,\big(\tilde{y}^{(h)}\circ\Psi^{(h)}\big)(hs,\xi,\zeta) - 
\frac{1}{h}\,\bar{Q}\Psi^{(h)}(hs,\xi,\zeta),
\end{equation*}
we obtain the following relation:
\begin{align}\label{trace}
&\int_{D}\Big|\,\big(\tilde{y}^{(h)}\circ\Psi^{(h)} - \bar{Q}\Psi^{(h)}\big)(0,\xi,\zeta) - \int_{D}\big(\tilde{y}^{(h)}\circ\Psi^{(h)} -\bar{Q}\Psi^{(h)}\big)(0,\xi,\zeta)\,d\xi\,d\zeta\,\Big|^2d\xi\,d\zeta \nonumber\\
&\leq c\,h \int_{(0,l h)\times D}\big|\,\nabla_h \big(\tilde{y}^{(h)}\circ\Psi^{(h)}\big) - \bar{Q}\nabla_h \Psi^{(h)}\,\big|^2 ds\,d\xi\,d\zeta.
\end{align}
Putting together (\ref{trace}) and (\ref{cil1}) we have, after easy computations,
\begin{align}\label{prima}
\int_{D}\Big|\,\big(\tilde{y}^{(h)}\circ\Psi^{(h)}\big)(0,\xi,\zeta) - \int_{D}\big(\tilde{y}^{(h)}\circ\Psi^{(h)}\big)(0,\xi,\zeta)\,d\xi\,d\zeta - h\,\bar{Q}(\xi\,\nu_2(0) + \zeta\,\nu_3(0))\Big|^2 d\xi\,d\zeta \leq c\,h^5.
\end{align}
In a similar way, if we define $\bar{\bar{Q}}:= Q^{(h)}(L)$, we deduce
\begin{align}\label{seconda}
\int_{D}\Big|\,\big(\tilde{y}^{(h)}\circ\Psi^{(h)}\big)(L,\xi,\zeta) - \int_{D}\big(\tilde{y}^{(h)}\circ\Psi^{(h)}\big)(L,\xi,\zeta)\,d\xi\,d\zeta - h\,\bar{\bar{Q}}(\xi\,\nu_2(L) + \zeta\,\nu_3(L))\Big|^2 d\xi\,d\zeta \leq c\,h^5.
\end{align} 
Now, subtracting (\ref{seconda}) from (\ref{prima}) and taking into account (\ref{bc}), we obtain
\begin{equation*}
 \int_{D}\big|\big[\bar{Q} - \bar{\bar{Q}}\big](\xi\,\nu_2(0) + \zeta\,\nu_3(0))\big|^2 d\xi\,d\zeta \leq c\,h^3,
\end{equation*}
which leads to 
\begin{equation}\label{bcQ}
 \big| Q^{(h)}(0) - Q^{(h)}(L)\big| \leq c\,h^{3/2}.
\end{equation}
If we define the sequences $\tilde{Q}^{(h)}$ and $R^{(h)}$ as in Theorem \ref{compactness}, it is easy to 
check that they also satisfy (\ref{bcQ}), hence (\ref{rota3}) is proved. For the estimates (\ref{rota}), 
(\ref{rota1}), and (\ref{rota2}) we proceed exactly as in Theorem \ref{compactness}.

Let us define the sequences $v^{(h)}$, $w^{(h)}$ and $u^{(h)}$ as in (\ref{vh}), (\ref{wh}) and (\ref{uh}). 
The convergence properties follow from Theorem \ref{compactness}. It remains only to verify the boundary 
conditions for the limiting functions $v$ and $w$.
Since $\Psi^{(h)}(0,\xi,\zeta) = \Psi^{(h)}(L,\xi,\zeta)$ and $Y^{(h)}(0,\xi,\zeta) = Y^{(h)}(L,\xi,\zeta)$ for 
every $(\xi,\zeta)\in D$, we have by definition that $v^{(h)}(0) = v^{(h)}(L)$ and $w^{(h)}(0) = w^{(h)}(L)$. 
Hence we directly obtain that $v$ and $w$ satisfy 
\begin{equation}\label{bcvw}
v(0) = v(L) \quad \hbox{and} \quad w(0) = w(L).
\end{equation}
Now notice that, by definition, 
\begin{equation}\label{impo}
\nabla_{h} Y^{(h)}\,\big(\nabla_{h}\Psi^{(h)}\big)^{-1} = 
\big(\bar{R}^{(h)}\big)^T \big(\nabla_{h} \tilde{y}^{(h)}\big)\circ \Psi^{(h)}.
\end{equation}
Therefore, using (\ref{impo}) and the fact that $\Psi^{(h)}(0,\xi,\zeta) = \Psi^{(h)}(L,\xi,\zeta)$ 
for every $(\xi,\zeta) \in D$, we have in particular that 
\begin{equation*}
\frac{\big(\nabla_{h} Y^{(h)}\,\big(\nabla_{h}\Psi^{(h)}\big)^{-1} - Id\big)}{h}(0,\xi,\zeta) = 
\frac{\big(\nabla_{h} Y^{(h)}\,\big(\nabla_{h}\Psi^{(h)}\big)^{-1} - Id\big)}{h}(L,\xi,\zeta)
\end{equation*}
for every $(\xi,\zeta) \in D$. The last relation, together with property (d), implies that $A(0) = A(L)$. 
Hence $v'(0) = v'(L)$ and so the proof is concluded.
\end{proof}
Now we are in a position to prove the $\Gamma$-convergence of the sequence $\big(\tilde{I}^{(h)}\big)/h^4$. As we have already noticed, the limit functional has the same expression as in (\ref{defI0}), but the class of deformations on which it is finite includes the boundary conditions. More precisely we have the following convergence result.

\begin{thm}\label{limite}

(1) Let $u,w\in W^{1,2}(0,L)$ and let $v\in W^{2,2}((0,L);\mathbb{R}^3)$ 
be such that $v'\cdot \tau = 0$. Assume also that $v$ and $w$ satisfy the 
boundary conditions (\ref{bcvw}). Then, for every positive sequence $(h_j)$
converging to zero and every sequence 
$\big(\tilde{y}^{(h_j)}\big)\subset
W^{1,2}(\widetilde{\Omega}_{h_j};\mathbb{R}^{3})$ such that the
sequence $Y^{(h_j)}:= \tilde{y}^{(h_j)}\circ \Psi^{(h_j)}$
satisfies the properties (a)-(d) of Theorem \ref{compbc}, it
turns out that
\begin{equation}\label{ineq}
\liminf_{j\rightarrow
\infty}\frac{1}{h_j^{6}}\int_{\widetilde{\Omega}_{h_j}}
W\Big(\big(\Psi^{(h_j)}\big)^{-1} (x),\nabla
\tilde{y}^{(h_j)}(x)\Big) dx \geq 
I_{3}(u, v, w),
\end{equation} 
where $I_{3}$ is defined in (\ref{defI0}).

(2) \,For every $u, w \in W^{1,2}(0,L)$ and $v\in
W^{2,2}((0,L);\mathbb{R}^3)$ satisfying the boundary conditions and such that $v'\cdot \tau = 0$, there
exists a sequence $\big(\check{y}^{(h)}\big)\subset
W^{1,2}(\tilde{\Omega}_{h}; \mathbb{M}^{3\times 3})$  such that, setting 
$Y^{(h)}:= \check{y}^{(h)}\circ \Psi^{(h)}$, we have
\vspace{.15cm}
\begin{itemize}
\item[(i)] $\big(\nabla_{h}
Y^{(h)}\big(\nabla_{h}\Psi^{(h)}\big)^{-1}- Id\big)/h
\rightarrow A$ \, strongly in $L^{2}(\Omega;\mathbb{M}^{3\times
3})$;
\vspace{.2cm}
\item[(ii)] $v^{(h)}\rightarrow v$ \, strongly in $W^{1,2}((0, L);\mathbb{R}^{3})$;
\vspace{.2cm}
\item[(iii)] \,$w^{(h)}\rightharpoonup w$\, weakly in $W^{1,2}(0, L)$;
\vspace{.2cm}
\item[(iv)] $u^{(h)}\rightharpoonup u$\, weakly in $W^{1,2}(0, L)$,
\end{itemize}
where $A$, $v^{(h)}$, $w^{(h)}$, and $u^{(h)}$ are
defined as in (\ref{A}), (\ref{vh}), (\ref{wh}) and (\ref{uh}). Moreover,
\begin{equation}\label{limsupfin}
\limsup_{h\rightarrow 0}\frac{1}{h^6}\int_{\widetilde{\Omega}_{h}}
W\Big(\big(\Psi^{(h)}\big)^{-1} (x),\nabla
\check{y}^{(h)}(x)\Big) dx \leq I_{3}(u, v, w),
\end{equation}
where $I_{3}$ is defined in (\ref{defI0}).
\end{thm}

\begin{proof}
(1) The proof of this part can be done repeating exactly the proof of Theorem \ref{scithm2}.\\
\noindent
(2) As in Theorem \ref{Gamcon3}, we preliminarly assume that $u, w\in C^{1}[0,L]$ and $v\in C^{2}([0,
L];\mathbb{R}^3)$. Let $\varphi \in C^1(\bar{\Omega};\mathbb{R}^3)$ and define $\beta: \Omega \rightarrow \mathbb{R}^3$ 
as $\beta(s,\xi,\zeta) := R_0(s)\varphi(s,\xi,\zeta)$.

\noindent
Let $\gamma_2$, $\gamma_3$ and $\kappa^{(h)}$ be as in the proof of Theorem \ref{Gamcon3}.
For every $h > 0$ let us consider a function $\vartheta^{(h)}\in C^1[0,L]$ supported in $[L- \sqrt{h},L]$, such that $\vartheta^{(h)}(L) = 1$ and $\big|\big(\vartheta^{(h)}\big)'\big| \leq \frac{c}{\sqrt{h}}$.
Then let us define the function $Y^{(h)}: \Omega\rightarrow \mathbb{R}^{3}$ as
\begin{align*}
Y^{(h)}&=\, \Psi^{(h)} + h\,v + h^2\, u\,\kappa^{(h)} + h^2 \,\frac{u(L) - u(0)}{L}\,\bigg(\int_0^s (L - \sigma)\,\tau'(\sigma)\,d\sigma - h\,(L - s)\,(\xi\,k_2 + \zeta\,k_3)\,\tau\bigg) \\ 
&+ h^2 \xi\,A\,\nu_2 +  h^2 \zeta\,A\,\nu_3 + h^3 \beta^{(h)},
\end{align*}
where $\beta^{(h)}(s,\xi,\zeta) := \beta(s,\xi,\zeta) + \vartheta^{(h)}(s)(\beta(0,\xi,\zeta) - 
\beta(L,\xi,\zeta))$. It turns out that the function $Y^{(h)}$ satisfies periodic boundary conditions in $(0,L)$. 
Indeed, 
\begin{align}\label{bcinzero}
Y^{(h)}(0,\xi,\zeta) &= \Psi^{(h)}(0,\xi,\zeta) + h\,v(0) + h^2\, u(0)\,\tau(0) - h^3 u(L)\,\tau(0)(\xi\,k_2(0) + \zeta\,k_3(0))+\nonumber\\
&+ h^2 \xi\,A(0)\,\nu_2(0) +  h^2 \zeta\,A(0)\,\nu_3(0) + h^3 \beta(0,\xi,\zeta),
\end{align}
and, using the assumptions (\ref{bc}) and (\ref{bcvw}), we have 
\begin{align}\label{bcinL}
&Y^{(h)}(L,\xi,\zeta)=\, \Psi^{(h)}(0,\xi,\zeta) + h\,v(0) + h^2\, u(L)\,\tau(0)\,(1 - h\,\xi\,k_2(0) - h\,\zeta\,k_3(0)) +\nonumber\\ 
&+ h^2 \,\frac{u(L) - u(0)}{L}\,\int_0^L (L - \sigma)\,\tau'(\sigma)\,d\sigma + h^2 \xi\,A(0)\,\nu_2(0) +  h^2 \zeta\,A(0)\,\nu_3(0) + h^3 \beta(0,\xi,\zeta). 
\end{align}
Now notice that, using $\tau = \gamma'$ and $\gamma(0) = \gamma(L)$, we have 
\begin{equation*}
\int_0^L (L - \sigma)\,\tau'(\sigma)\,d\sigma = -L\,\tau(0) + \int_0^L \tau(\sigma)\,d\sigma = -L \,\tau(0).
\end{equation*}
Plugging this equality into (\ref{bcinL}) we obtain
\begin{align}\label{bcinL2}
&Y^{(h)}(L,\xi,\zeta)=\, \Psi^{(h)}(0,\xi,\zeta) + h\,v(0) + h^2\, u(L)\,\tau(0)\,(1 - h\,\xi\,k_2(0) - h\,\zeta\,k_3(0)) +\nonumber\\ 
&- h^2 \,(u(L) - u(0))\,\tau(0) + h^2 \xi\,A(0)\,\nu_2(0) +  h^2 \zeta\,A(0)\,\nu_3(0) + h^3 \beta(0,\xi,\zeta) 
\end{align}
which is the same expression in (\ref{bcinzero}).

Moreover, the convergence properties (i)-(iv) can be deduced as in Theorem \ref{bfa}. 
For the scaled gradient of $Y^{(h)}$ we have, using (\ref{antisim})
\begin{align}\label{gradYhfin}
&\nabla_{h}Y^{(h)} =\, \nabla_{h}\Psi^{(h)} + h\,A\,R_0 + h^2\,\Bigg(\big(A\,R_0\big)'\Bigg(\begin{array}{c}
0\\
\xi\\
\zeta
\end{array}\Bigg) + \,u'\,\tau \,\bigg|\,
\partial_{\xi} \beta^{(h)}\,\bigg|\,
\partial_{\zeta} \beta^{(h)}
 \Bigg) \,+ \nonumber\\  
& + h^2\,\bigg(u + \frac{u(L) - u(0)}{L}\,(L-s)\bigg)\,\big(\tau'\otimes\tau - \tau\otimes\tau'\big)\,R_0 + h^{3}\partial_s \beta^{(h)}\otimes e_1 + O(h^3), \nonumber
\end{align}
where $\big|\partial_s \beta^{(h)}\big| \leq \frac{c}{\sqrt{h}}$. Now, since $\beta^{(h)} \rightarrow \beta$,
$\partial_{\xi}\beta^{(h)}\rightarrow \partial_{\xi}\beta$ and $\partial_{\zeta}\beta^{(h)}\rightarrow \partial_{\zeta}\beta$ strongly in $L^2(\Omega;\mathbb{R}^3)$, one can prove a convergence result like 
(\ref{quasifin}) and obtain the general case by approximation. 
\end{proof}


\bigskip
\bigskip
\centerline{\textsc{Acknowledgements}}
\bigskip
\noindent
I warmly thank Maria Giovanna Mora for having proposed
to me the study of this problem and for many helpful and interesting
suggestions. I would like also to thank Gianni Dal Maso for several
stimulating discussions on the subject of this paper. 
Moreover I express my gratitude to the anonymous referee for his  
valuable comments.


\bigskip

\addcontentsline{toc}{chapter}{References}

\begin{thebibliography}{9}

\bibitem[1]{ABP91} Acerbi E., Buttazzo G., Percivale D.: A variational definition of the strain 
energy for an elastic string. \textit{J. Elasticity}, \textbf{25} (1991), 137--148. 

\bibitem[2]{Cia2} Ciarlet P.G., Mathematical elasticity II - theory of plates. Elsevier, Amsterdam, 1997.

\bibitem[3]{DM93} Dal Maso G.: An introduction to $\Gamma$-convergence. Birkh\"auser, Boston, 1993.

\bibitem[4]{FJM02} Friesecke G., James R.D., M\"uller S.: A theorem on geometric rigidity and the derivation of    nonlinear plate theory from three-dimensional elasticity. \textit{Comm. Pure Appl. Math.}, \textbf{55} (2002), 1461--1506. 

\bibitem[5]{FJMM03} Friesecke G., James R.D., Mora M.G., M\"uller S.: Derivation of nonlinear bending theory for shells from three-dimensional nonlinear elasticity by $\Gamma$-convergence. \textit{C. R. Math. Acad. Sci. Paris},  \textbf{336} (2003), 697--702. 

\bibitem[6]{FJM06} Friesecke G., James R.D., M\"uller S.: A hierarchy of plate models derived from nonlinear 
elasticity by $\Gamma$-convergence. \textit{Arch. Ration. Mech. Anal.}, \textbf{180}/2 (2006), 183--236.

\bibitem[7]{Griso} Griso G.: Asymptotic behaviour of curved rods by the unfolding method. \textit{Math. Methods Appl. Sci.}, \textbf{27} (2004), 2081--2110.

\bibitem[8]{JuTam} Jurak M., Tambaca J.: Linear curved rod model. General curve. \textit{Math. Models Methods Appl. Sci.}, \textbf{11}/7 (2001), 1237--1252.

\bibitem[9]{LDR95} Le Dret H., Raoult A.: The nonlinear membrane model as variational limit of nonlinear three-dimensional elasticity. \textit{J. Math. Pures Appl. (9)}, \textbf{74} (1995), 549--578.

\bibitem[10]{LDR00} Le Dret H., Raoult A.: The membrane shell model in nonlinear elasticity: a variational asymptotic derivation. \textit{In ``Mechanics: from theory to computation''}, \textit{Springer, New York}, 2000.

\bibitem[11]{MM03} Mora M.G., M\"uller S.: Derivation of the nonlinear bending-torsion theory 
for inextensible rods by $\Gamma$-convergence. \textit{Calc. Var. Partial Differential Equations}, \textbf{18} (2003), 287--305. 

\bibitem[12]{MGMM04} Mora M.G., M\"uller S.: A nonlinear model for inextensible rods as a low energy 
$\Gamma$-limit of three-dimensional nonlinear elasticity. \textit{Ann. Inst. H. Poincar\'e Anal. Non Lin\'eaire},  \textbf{21} (2004), 271--293. 

\bibitem[13]{MGMS06} Mora M.G., M\"uller S., Schultz M.G.: Convergence of equilibria of planar thin elastic beams. 
\textit{Indiana Univ. Math. J.}, (2006).

\bibitem[14]{P02} Pantz O.: Le mod\`ele de poutre inextentionnelle comme limite de l' \'elasticit\'e non-lin\'eaire tridimensionnelle. Preprint (2002).

\bibitem[15]{S06} Scardia L.: The nonlinear bending-torsion theory for curved rods as $\Gamma$-limit of three-dimensional elasticity. \textit{Asymptot. Anal.}, \textbf{47}/3-4 (2006), 317--343.

\bibitem[16]{Sepp1} Seppecher P., Pideri C.: Asymptotics of a non-planar beam in linear elasticity. Preprint (2005).

\bibitem[17]{Sepp2} Seppecher P., Pideri C.: Asymptotics of a non-planar beam in non linear elasticity. \textit{Asymptot. Anal.}, \textbf{48}/1-2 (2006), 33--54.

\end{thebibliography}
\end{document}